\documentclass{lmcs}
\pdfoutput=1

\usepackage{lastpage}
\lmcsdoi{16}{1}{4}
\lmcsheading{}{\pageref{LastPage}}{}{}%
{Jan.~11,~2019}{Jan.~17,~2020}{}

\usepackage{amsmath,amssymb,stmaryrd,xspace}
\usepackage{amsthm}
\usepackage{tikz-cd}
\usepackage{multirow}
\usepackage{mathtools}
\usepackage{bm}
\usepackage{hyperref}
\usepackage{tikz,xy}
\newcommand{\para}[1]{\subsection*{#1}}

\newcommand{\sta}{\psi}
\newcommand{\stb}{\phi}
\newcommand{\stc}{\omega}
\newcommand{\stl}{\xi}

\newcommand{\diagfont}[1]{{#1}}

\newcommand{\deff}[1]{\emph{#1}}

\newcommand{\strongpurification}{strong purification} 

\newcommand{\Rplus}{\mathbb{R}^+} 
\newcommand{\Rpos}{\Rplus} 

\newcommand{\notetoself}[1]{}
\newcommand{\omitfornow}[1]{}

\newcommand{\cat}[1]{\ensuremath{\mathbf{#1}}\xspace}

\newcommand{\HilbP}{\Hilb_\sim}
\newcommand{\FHilbP}{\FHilb_\sim}

\newcommand{\MatS}{\Mat_S}

\newcommand{\Quant}[1]{\cat{Quant}_{#1}}

\newcommand{\Class}{\cat{Class}}

\newcommand{\catA}{\cat{A}}
\newcommand{\catC}{\cat{C}}
\newcommand{\catH}{\cat{H}}

\newcommand{\catB}{\cat{B}}

\newcommand{\catD}{\cat{D}}

\newcommand{\FHilb}{\cat{FHilb}}
\newcommand{\Hilb}{\cat{Hilb}}
\newcommand{\Rel}{\cat{Rel}}

\newcommand{\CPM}{\ensuremath{\mathsf{CPM}}\xspace}

\newcommand{\pos}{\mathrm{pos}} 

\newcommand{\pure}{\mathsf{pure}}
\newcommand{\prepure}{\mathsf{p}} 
\newcommand{\selfadj}{\mathsf{sa}} 

\newcommand{\id}[1]{\ensuremath{\mathrm{id}_{#1}}}

\newcommand{\op}{\ensuremath{\mathrm{\rm op}}}

\newcommand{\discard}[1]{\ensuremath{\tinygroundnew_{#1}}}

\newcommand{\supp}{\mathsf{supp}}
\newcommand{\dring}[1]{D(#1)}

\DeclareMathOperator{\Tr}{Tr}
\newcommand{\Dbl}[1]{\widehat{#1}}

\newcommand{\ket}[1]{\ensuremath{| #1 \rangle}}

\newcommand{\Ker}{\mathrm{Ker}}
\newcommand{\Coker}{\mathrm{Coker}}

\newcommand{\img}{\mathrm{im}}
\newcommand{\Img}{\mathrm{Im}}
\newcommand{\CoIm}{\mathrm{Coim}}
\newcommand{\coim}{\mathrm{coim}}
\newcommand{\coker}{\mathrm{coker}}

\newcommand{\coproj}{\kappa}
\newcommand{\biprod}{\oplus}

\newcommand{\pbiprod}{\mathbin{\dot{\oplus}}}

\newcommand{\pcoproj}{\coproj}
\newcommand{\pproj}{\pi}

\newcommand{\GP}{\mathsf{GP}}
\newcommand{\plusI}[1]{\mathsf{GP}({#1})}
\newcommand{\plusIdagsym}{\mathsf{GP}^{\dagger}}

\newcommand{\obb}[1]{\mathbf{#1}}

\newcommand{\Mat}{\mathsf{Mat}}

\newcommand{\hilbH}{\mathcal{H}} 
\newcommand{\hilbK}{\mathcal{K}} 

\usetikzlibrary{decorations.pathreplacing,decorations.markings,arrows.meta,backgrounds}
\pgfdeclarelayer{edgelayer}
\pgfdeclarelayer{nodelayer}
\pgfsetlayers{background,edgelayer,nodelayer,main}

\tikzstyle{cdot}=[circle, draw=black, fill=black!25, inner sep=.4ex] 
\tikzstyle{bigdot}=[dot, inner sep=0pt]
\tikzstyle{whitedot}=[circle, draw=black, fill=white, inner sep=.4ex]
\tikzstyle{blackdot}=[circle, draw=black, fill=black, inner sep=.4ex]
\tikzset{arrow/.style={decoration={
    markings,
    mark=at position #1 with \arrow{>[length=2pt, width=3pt]}},
    postaction=decorate},
    reverse arrow/.style={decoration={
    markings,
    mark=at position #1 with {{\arrow{<[length=2pt, width=3pt]}}}},
    postaction=decorate}
}

\newcommand{\tinymultflip}[1][cdot]{
\smash{\raisebox{-2pt}{\hspace{-5pt}\ensuremath{\begin{pic}[scale=0.4,yscale=1]
    \node (0) at (0,0) {};
    \node[#1, inner sep=1.5pt] (1) at (0,0.55) {};
    \node (2) at (-0.5,1) {};
    \node (3) at (0.5,1) {};
    \draw (0.center) to (1.center);
    \draw (1.center) to [out=left, in=down, out looseness=1.5] (2.center);
    \draw (1.center) to [out=right, in=down, out looseness=1.5] (3.center);
    \node[#1, inner sep=1.5pt] (1) at (0,0.55) {};
\end{pic}
}\hspace{-3pt}}}}

\newcommand{\tinycounit}[1][cdot]{
\smash{\raisebox{-3pt}{\hspace{-3pt}\ensuremath{\begin{pic}[scale=0.4,yscale=1]
    \node (0) at (0,0) {};
    \node[#1, inner sep=1.5pt] (1) at (0,0.55) {};
    \draw (0.center) to (1.north);
        \node[#1, inner sep=1.5pt] (1) at (0,0.55) {};
\end{pic}
}\hspace{-1pt}}}}

\newcommand{\tinycup}{\smash{\raisebox{2pt}{\hspace{-2pt}\ensuremath{\begin{pic}[scale=0.2]
   \pgftransformscale{1.5} \draw[arrow=.6, scale = 1] (0,0) to[out=-90,in=-90,looseness=1.5] (1.5,0);
\end{pic}}}}}

\newcommand{\tinycap}{\smash{\raisebox{-3pt}{\hspace{-2pt}\ensuremath{\begin{pic}[scale=0.2, yscale=-1]
   \pgftransformscale{1.5} \draw[arrow=.6, scale = 1] (0,0) to[out=-90,in=-90,looseness=1.5] (1.5,0);
\end{pic}}}}}

\newenvironment{pic}[1][] {\begin{aligned}\begin{tikzpicture}[scale=2.0, font=\tiny,#1]}{\end{tikzpicture}\end{aligned}} 

\newif\ifvflip\pgfkeys{/tikz/vflip/.is if=vflip}
\newif\ifhflip\pgfkeys{/tikz/hflip/.is if=hflip}
\newif\ifhvflip\pgfkeys{/tikz/hvflip/.is if=hvflip}

\newenvironment{picc}[1][]
{\begin{aligned}\begin{tikzpicture}[font=\tiny,#1]}
{\end{tikzpicture}\end{aligned}}

\newlength\minimummorphismwidth
\newlength\stateheight
\newlength\minimumstatewidth
\newlength\connectheight
\tikzset{colour/.initial=white}

\tikzstyle{pure}=[line width=.7pt]

\makeatletter

\pgfdeclareshape{groundd}
{
    \savedanchor\centerpoint
    {
        \pgf@x=0pt
        \pgf@y=0pt
    }
    \anchor{center}{\centerpoint}
    \anchorborder{\centerpoint}

    \anchor{north}
    {
        \pgf@x=0pt
        \pgf@y=0.16\stateheight
    }
    \anchor{south}
    {
        \pgf@x=0pt
        \pgf@y=0pt
    }
    \saveddimen\overallwidth
    {
        \pgfkeysgetvalue{/pgf/minimum width}{\minwidth}
        \pgf@x=\minimumstatewidth
        \ifdim\pgf@x<\minwidth
            \pgf@x=\minwidth
        \fi
    }
    \backgroundpath
    {
        \begin{pgfonlayer}{main} 
        \pgfsetstrokecolor{black}
        \pgfsetlinewidth{1.25pt}
        \ifhflip
            \pgftransformyscale{-1}
        \fi
        \pgftransformscale{0.5}
        \pgfpathmoveto{\pgfpoint{-0.5*\overallwidth}{0pt}}
        \pgfpathlineto{\pgfpoint{0.5*\overallwidth}{0pt}}
        \pgfpathmoveto{\pgfpoint{-0.33*\overallwidth}{0.33*\stateheight}}
        \pgfpathlineto{\pgfpoint{0.33*\overallwidth}{0.33*\stateheight}}
        \pgfpathmoveto{\pgfpoint{-0.16*\overallwidth}{0.66*\stateheight}}
        \pgfpathlineto{\pgfpoint{0.16*\overallwidth}{0.66*\stateheight}}
        \pgfpathmoveto{\pgfpoint{-0.02*\overallwidth}{\stateheight}}
        \pgfpathlineto{\pgfpoint{0.02*\overallwidth}{\stateheight}}
        \pgfusepath{stroke}
        \end{pgfonlayer}
    }
}

\newcommand{\tinygroundnew}{
\smash{
\raisebox{-1pt}
{\hspace{-3pt}
\ensuremath{
\begin{picc}[scale=1.0] 
    \node[upground, xscale=0.8, yscale=0.7] (1) at (0,0.16) {};
    \draw (0,0.03) to (0,-0.25);
\end{picc}
}\hspace{-1pt}}}}

\usetikzlibrary{decorations.pathreplacing,decorations.markings,arrows.meta,backgrounds,shapes}
\usetikzlibrary{circuits.ee.IEC}
\pgfdeclarelayer{edgelayer}
\pgfdeclarelayer{nodelayer}
\pgfsetlayers{background,edgelayer,nodelayer,main}
\tikzstyle{none}=[inner sep=0mm]
\tikzstyle{every loop}=[]
\tikzstyle{mark coordinate}=[inner sep=0pt,outer sep=0pt,minimum size=3pt,fill=black,circle]

\tikzset{arrow/.style={decoration={
    markings,
    mark=at position #1 with \arrow{>[length=2pt, width=3pt]}},
    postaction=decorate},
    reverse arrow/.style={decoration={
    markings,
    mark=at position #1 with {{\arrow{<[length=2pt, width=3pt]}}}},
    postaction=decorate}
}

\tikzset{
    roundbox/.style={
           rectangle,
           rounded corners,
           draw=black, 
           fill=white,
           text width=8em,
           minimum height=2em,
           text centered},
    bigedge/.style={
           ->,
           shorten <=2pt,
           shorten >=2pt,},
     incl/.style={
           left hook->,
           shorten <=2pt,
           shorten >=2pt,},
     incl2/.style={
           right hook->,
           shorten <=2pt,
           shorten >=2pt,},
}

\tikzstyle{upground}=[circuit ee IEC,thick,ground,rotate=90,scale=1.5]
\tikzstyle{downground}=[circuit ee IEC,thick,ground,rotate=-90,scale=1.5]
\tikzstyle{env}=[copoint,regular polygon rotate=0,minimum width=0.2cm, fill=black]
\tikzstyle{bigground}=[regular polygon,regular polygon sides=3,draw=gray,scale=0.50,inner sep=-0.5pt,minimum width=10mm,fill=gray]
\tikzstyle{probs}=[shape=semicircle,fill=white,draw=black,shape border rotate=180,minimum width=1.2cm]

\tikzstyle{box}=[draw,shape=rectangle,inner sep=2pt,minimum height=\mapminh,minimum width=6mm,fill=white] 
\tikzstyle{medium box}=[draw,shape=rectangle,inner sep=2pt,minimum height=\mapminh,minimum width=10mm,fill=white] 
\tikzstyle{dot}=[inner sep=0mm,minimum width=2mm,minimum height=2mm,draw,shape=circle]  
\tikzstyle{black dot}=[dot,fill=black]
\tikzstyle{white dot}=[dot,fill=white,,text depth=-0.2mm]
\tikzstyle{grey dot}=[dot,fill=black!25] 

\tikzstyle{corner1}=[box,fill=white, font=\footnotesize] %
\tikzstyle{corner2}=[dot,fill=white, font=\footnotesize] %
\tikzstyle{corner3}=[dot,fill=black!25, font=\footnotesize] %
\tikzstyle{corner4}=[dot,fill=black, font=\footnotesize] %

\tikzstyle{scalar}=[circle,draw,inner sep=1pt] 

\tikzstyle{dscalar}=[diamond,doubled, draw,inner sep=0.5pt,font=\small]

\tikzstyle{point}=[regular polygon,regular polygon sides=3,draw,scale=0.75,inner sep=-0.5pt,minimum width=9mm,fill=white,regular polygon rotate=180]
\tikzstyle{small dot}=[inner sep=0.5mm,minimum width=0pt,minimum height=0pt,draw,shape=circle]
\tikzstyle{small black dot}=[small dot,fill=black]
\tikzstyle{small white dot}=[small dot,fill=white]

\tikzstyle{copoint}=[regular polygon,regular polygon sides=3,draw,scale=0.75,inner sep=-0.5pt,minimum width=9mm,fill=white]
\tikzstyle{dpoint}=[point,doubled]
\tikzstyle{dcopoint}=[copoint,doubled]

\tikzstyle{wide copoint}=[fill=white,draw,shape=isosceles triangle,shape border rotate=90,isosceles triangle stretches=true,inner sep=0pt,minimum width=1.5cm,minimum height=6.12mm]
\tikzstyle{point}=[kpoint]
\tikzstyle{dagpoint}=[kpoint]
\tikzstyle{dagpointadj}=[kpointadj]
\tikzstyle{wide point}=[wide kpoint]
\tikzstyle{medium dagmap}=[medium map]
\tikzstyle{semilarge dagmap}=[semilarge map]

\tikzstyle{wide dpoint}=[fill=white,doubled,draw,shape=isosceles triangle,shape border rotate=-90,isosceles triangle stretches=true,inner sep=0pt,minimum width=1.5cm,minimum height=6.12mm,yshift=-0.0mm]

\tikzstyle{tinypoint}=[regular polygon,regular polygon sides=3,draw,scale=0.55,inner sep=-0.15pt,minimum width=6mm,fill=white,regular polygon rotate=180]

\makeatletter
\newcommand{\boxshape}[3]{%
\pgfdeclareshape{#1}{
\inheritsavedanchors[from=rectangle] 
\inheritanchorborder[from=rectangle]
\inheritanchor[from=rectangle]{center}
\inheritanchor[from=rectangle]{north}
\inheritanchor[from=rectangle]{south}
\inheritanchor[from=rectangle]{west}
\inheritanchor[from=rectangle]{east}
\backgroundpath{
\southwest \pgf@xa=\pgf@x \pgf@ya=\pgf@y
\northeast \pgf@xb=\pgf@x \pgf@yb=\pgf@y

\@tempdima=#2
\@tempdimb=#3

\pgfpathmoveto{\pgfpoint{\pgf@xa - 5pt + \@tempdima}{\pgf@ya}}
\pgfpathlineto{\pgfpoint{\pgf@xa - 5pt - \@tempdima}{\pgf@yb}}
\pgfpathlineto{\pgfpoint{\pgf@xb + 5pt + \@tempdimb}{\pgf@yb}}
\pgfpathlineto{\pgfpoint{\pgf@xb + 5pt - \@tempdimb}{\pgf@ya}}
\pgfpathlineto{\pgfpoint{\pgf@xa - 5pt + \@tempdima}{\pgf@ya}}
\pgfpathclose
}
}}

\boxshape{NEbox}{0pt}{3pt} 
\boxshape{SEbox}{0pt}{-3pt}
\boxshape{NWbox}{3pt}{0pt}
\boxshape{SWbox}{-3pt}{0pt}
\makeatother

\tikzstyle{cloud}=[shape=cloud,draw,minimum width=1.5cm,minimum height=1.5cm]

\newcommand{\mapminh}{5mm} 
\newcommand{\maplw}{0.7pt} 

\tikzstyle{map}=[draw,shape=NEbox,inner sep=2pt,minimum height=\mapminh,fill=white, line width = \maplw] %
\tikzstyle{dagmap}=[map]
\tikzstyle{dashedmap}=[draw,dashed,shape=NEbox,inner sep=2pt,minimum height=\mapminh,fill=white, line width = \maplw]
\tikzstyle{mapdag}=[draw,shape=SEbox,inner sep=2pt,minimum height=\mapminh,fill=white, line width = \maplw]
\tikzstyle{mapadj}=[draw,shape=SEbox,inner sep=2pt,minimum height=\mapminh,fill=white, line width = \maplw]
\tikzstyle{maptrans}=[draw,shape=SWbox,inner sep=2pt,minimum height=\mapminh,fill=white, line width = \maplw]
\tikzstyle{mapconj}=[draw,shape=NWbox,inner sep=2pt,minimum height=\mapminh,fill=white, line width = \maplw]

\tikzstyle{medium map}=[draw,shape=NEbox,inner sep=2pt,minimum height=\mapminh,fill=white,minimum width=7mm, line width = \maplw]
\tikzstyle{medium map dag}=[draw,shape=SEbox,inner sep=2pt,minimum height=\mapminh,fill=white,minimum width=7mm, line width = \maplw]
\tikzstyle{medium map adj}=[draw,shape=SEbox,inner sep=2pt,minimum height=\mapminh,fill=white,minimum width=7mm, line width = \maplw]
\tikzstyle{medium map trans}=[draw,shape=SWbox,inner sep=2pt,minimum height=\mapminh,fill=white,minimum width=7mm, line width = \maplw]
\tikzstyle{medium map conj}=[draw,shape=NWbox,inner sep=2pt,minimum height=\mapminh,fill=white,minimum width=7mm, line width = \maplw]
\tikzstyle{semilarge map}=[draw,shape=NEbox,inner sep=2pt,minimum height=\mapminh,fill=white,minimum width=9.5mm, line width = \maplw]
\tikzstyle{semilarge map trans}=[draw,shape=SWbox,inner sep=2pt,minimum height=\mapminh,fill=white,minimum width=9.5mm, line width = \maplw]
\tikzstyle{semilarge map adj}=[draw,shape=SEbox,inner sep=2pt,minimum height=\mapminh,fill=white,minimum width=9.5mm, line width = \maplw]
\tikzstyle{semilarge map dag}=[draw,shape=SEbox,inner sep=2pt,minimum height=\mapminh,fill=white,minimum width=9.5mm, line width = \maplw]
\tikzstyle{semilarge map conj}=[draw,shape=NWbox,inner sep=2pt,minimum height=\mapminh,fill=white,minimum width=9.5mm, line width = \maplw]
\tikzstyle{large map}=[draw,shape=NEbox,inner sep=2pt,minimum height=\mapminh,fill=white,minimum width=12mm, line width = \maplw]
\tikzstyle{large map conj}=[draw,shape=NWbox,inner sep=2pt,minimum height=\mapminh,fill=white,minimum width=12mm, line width = \maplw]
\tikzstyle{very large map}=[draw,shape=NEbox,inner sep=2pt,minimum height=\mapminh,fill=white,minimum width=17mm, line width = \maplw]

\makeatletter

\pgfdeclareshape{cornerpoint}{
\inheritsavedanchors[from=rectangle] 
\inheritanchorborder[from=rectangle]
\inheritanchor[from=rectangle]{center}
\inheritanchor[from=rectangle]{north}
\inheritanchor[from=rectangle]{south}
\inheritanchor[from=rectangle]{west}
\inheritanchor[from=rectangle]{east}
\backgroundpath{
\southwest \pgf@xa=\pgf@x \pgf@ya=\pgf@y
\northeast \pgf@xb=\pgf@x \pgf@yb=\pgf@y

\pgfmathsetmacro{\pgf@shorten@left}{\pgfkeysvalueof{/tikz/shorten left}}
\pgfmathsetmacro{\pgf@shorten@right}{\pgfkeysvalueof{/tikz/shorten right}}

\pgfpathmoveto{\pgfpoint{0.5 * (\pgf@xa + \pgf@xb)}{\pgf@ya - 5pt}}
\pgfpathlineto{\pgfpoint{\pgf@xa - 8pt + \pgf@shorten@left}{\pgf@yb - 1.5 * \pgf@shorten@left}}
\pgfpathlineto{\pgfpoint{\pgf@xa - 8pt + \pgf@shorten@left}{\pgf@yb}}
\pgfpathlineto{\pgfpoint{\pgf@xb + 8pt - \pgf@shorten@right}{\pgf@yb}}
\pgfpathlineto{\pgfpoint{\pgf@xb + 8pt - \pgf@shorten@right}{\pgf@yb - 1.5 * \pgf@shorten@right}}
\pgfpathclose
}
}

\pgfdeclareshape{cornercopoint}{
\inheritsavedanchors[from=rectangle] 
\inheritanchorborder[from=rectangle]
\inheritanchor[from=rectangle]{center}
\inheritanchor[from=rectangle]{north}
\inheritanchor[from=rectangle]{south}
\inheritanchor[from=rectangle]{west}
\inheritanchor[from=rectangle]{east}
\backgroundpath{
\southwest \pgf@xa=\pgf@x \pgf@ya=\pgf@y
\northeast \pgf@xb=\pgf@x \pgf@yb=\pgf@y

\pgfmathsetmacro{\pgf@shorten@left}{\pgfkeysvalueof{/tikz/shorten left}}
\pgfmathsetmacro{\pgf@shorten@right}{\pgfkeysvalueof{/tikz/shorten right}}

\pgfpathmoveto{\pgfpoint{0.5 * (\pgf@xa + \pgf@xb)}{\pgf@yb + 5pt}}
\pgfpathlineto{\pgfpoint{\pgf@xa - 8pt + \pgf@shorten@left}{\pgf@ya + 1.5 * \pgf@shorten@left}}
\pgfpathlineto{\pgfpoint{\pgf@xa - 8pt + \pgf@shorten@left}{\pgf@ya}}
\pgfpathlineto{\pgfpoint{\pgf@xb + 8pt - \pgf@shorten@right}{\pgf@ya}}
\pgfpathlineto{\pgfpoint{\pgf@xb + 8pt - \pgf@shorten@right}{\pgf@ya + 1.5 * \pgf@shorten@right}}
\pgfpathclose
}
}

\makeatother

\pgfkeyssetvalue{/tikz/shorten left}{0pt}
\pgfkeyssetvalue{/tikz/shorten right}{0pt}

\tikzstyle{kpoint common}=[draw,fill=white,inner sep=1pt, line width = \maplw, minimum height = 4mm, yshift=1.2pt] 
\tikzstyle{kpoint sc}=[shape=cornerpoint,kpoint common]
\tikzstyle{kpoint adjoint sc}=[shape=cornercopoint,kpoint common]
\tikzstyle{kpoint}=[shape=cornerpoint,shorten left=4pt,kpoint common]
\tikzstyle{kpoint adjoint}=[shape=cornercopoint,shorten left=5pt,kpoint common]
\tikzstyle{kpoint conjugate}=[shape=cornerpoint,shorten right=5pt,kpoint common]
\tikzstyle{kpoint transpose}=[shape=cornercopoint,shorten right=5pt,kpoint common]
\tikzstyle{kpoint symm}=[shape=cornerpoint,shorten left=5pt,shorten right=5pt,kpoint common]
\tikzstyle{kpointdag}=[kpoint adjoint]
\tikzstyle{kpointadj}=[kpoint adjoint]
\tikzstyle{kpointconj}=[kpoint conjugate]
\tikzstyle{kpointtrans}=[kpoint transpose]

\tikzstyle{big kpoint}=[kpoint, minimum width=1.2 cm, minimum height=8mm, inner sep=4pt, text depth=3mm]

\tikzstyle{wide kpoint}=[kpoint, minimum width=1 cm, inner sep=2pt]
\tikzstyle{wide kpointdag}=[kpointdag, minimum width=1 cm, inner sep=2pt]
\tikzstyle{wide kpointconj}=[kpointconj, minimum width=1 cm, inner sep=2pt]
\tikzstyle{wide kpointtrans}=[kpointtrans, minimum width=1 cm, inner sep=2pt]

\tikzstyle{gray kpoint}=[kpoint,fill=gray!50!white]
\tikzstyle{gray kpointdag}=[kpointdag,fill=gray!50!white]
\tikzstyle{gray kpointadj}=[kpointadj,fill=gray!50!white]
\tikzstyle{gray kpointconj}=[kpointconj,fill=gray!50!white]
\tikzstyle{gray kpointtrans}=[kpointtrans,fill=gray!50!white]

\tikzstyle{gray dkpoint}=[kpoint,fill=gray!50!white,doubled]
\tikzstyle{gray dkpointdag}=[kpointdag,fill=gray!50!white,doubled]
\tikzstyle{gray dkpointadj}=[kpointadj,fill=gray!50!white,doubled]
\tikzstyle{gray dkpointconj}=[kpointconj,fill=gray!50!white,doubled]
\tikzstyle{gray dkpointtrans}=[kpointtrans,fill=gray!50!white,doubled]

\tikzstyle{white label}=[draw,fill=white,rectangle,inner sep=0.7 mm]
\tikzstyle{gray label}=[draw,fill=gray!50!white,rectangle,inner sep=0.7 mm]
\tikzstyle{black label}=[draw,fill=black,rectangle,inner sep=0.7 mm]

\tikzstyle{dkpoint}=[kpoint,doubled]
\tikzstyle{wide dkpoint}=[wide kpoint,doubled]
\tikzstyle{dkpointdag}=[kpoint adjoint,doubled]
\tikzstyle{wide dkpointdag}=[wide kpointdag,doubled]
\tikzstyle{dkcopoint}=[kpoint adjoint,doubled]
\tikzstyle{dkpointadj}=[kpoint adjoint,doubled]
\tikzstyle{dkpointconj}=[kpoint conjugate,doubled]
\tikzstyle{dkpointtrans}=[kpoint transpose,doubled]

\tikzstyle{kscalar}=[kpoint common, shape=EBox, inner xsep=-1pt, inner ysep=3pt,font=\small]
\tikzstyle{kscalarconj}=[kpoint common, shape=WBox, inner xsep=-1pt, inner ysep=3pt,font=\small]

\tikzstyle{spekpoint}=[kpoint sc,minimum height=5mm,inner sep=3pt]
\tikzstyle{spekcopoint}=[kpoint adjoint sc,minimum height=5mm,inner sep=3pt]

\tikzstyle{dspekpoint}=[spekpoint,doubled]
\tikzstyle{dspekcopoint}=[spekcopoint,doubled]

\newif\ifvflip\pgfkeys{/tikz/vflip/.is if=vflip}
\newif\ifhflip\pgfkeys{/tikz/hflip/.is if=hflip}
\newif\ifhvflip\pgfkeys{/tikz/hvflip/.is if=hvflip}
\newlength\morphismheight
\newlength\wedgewidth
\tikzset{width/.initial=1mm}
\makeatletter
\tikzstyle{morphism}=[font=\small,morphismshape] 
\tikzstyle{morphismflip}=[morphism, hflip]

\pgfdeclareshape{morphismshape}
{
    \savedanchor\centerpoint
    {
        \pgf@x=0pt
        \pgf@y=0pt
    }
    \anchor{center}{\centerpoint}
    \anchorborder{\centerpoint}
    \saveddimen\overallwidth
    {
        \pgfkeysgetvalue{/tikz/width}{\minwidth}
        \pgf@x=\wd\pgfnodeparttextbox
        \ifdim\pgf@x<\minwidth
            \pgf@x=\minwidth
        \fi
    }
    \savedanchor{\upperrightcorner}
    {
        \pgf@y=.5\ht\pgfnodeparttextbox
        \advance\pgf@y by -.5\dp\pgfnodeparttextbox
        \pgf@x=.5\wd\pgfnodeparttextbox
    }
    \anchor{north}
    {
        \pgf@x=0pt
        \pgf@y=0.5\morphismheight
    }
    \anchor{north east}
    {
        \pgf@x=\overallwidth
        \multiply \pgf@x by 2
        \divide \pgf@x by 5
        \pgf@y=0.5\morphismheight
    }
    \anchor{east}
    {
        \pgf@x=\overallwidth
        \divide \pgf@x by 2
        \advance \pgf@x by 5pt
        \pgf@y=0pt
    }
    \anchor{west}
    {
        \pgf@x=-\overallwidth
        \divide \pgf@x by 2
        \advance \pgf@x by -5pt
        \pgf@y=0pt
    }
    \anchor{north west}
    {
        \pgf@x=-\overallwidth
        \multiply \pgf@x by 2
        \divide \pgf@x by 5
        \pgf@y=0.5\morphismheight
    }
    \anchor{south east}
    {
        \pgf@x=\overallwidth
        \multiply \pgf@x by 2
        \divide \pgf@x by 5
        \pgf@y=-0.5\morphismheight
    }
    \anchor{south west}
    {
        \pgf@x=-\overallwidth
        \multiply \pgf@x by 2
        \divide \pgf@x by 5
        \pgf@y=-0.5\morphismheight
    }
    \anchor{south}
    {
        \pgf@x=0pt
        \pgf@y=-0.5\morphismheight
    }
    \anchor{text}
    {
        \upperrightcorner
        \pgf@x=-\pgf@x
        \pgf@y=-\pgf@y
    }
    \backgroundpath
    {
    \begin{scope}
        \pgfkeysgetvalue{/tikz/fill}{\morphismfill}
        \pgfsetstrokecolor{black}
        \pgfsetlinewidth{.7pt}
        \begin{scope}
        \pgfsetstrokecolor{black}
        \pgfsetfillcolor{white}
        \pgfsetlinewidth{.7pt}
                \ifhflip
                    \pgftransformyscale{-1}
                \fi
                \ifvflip
                    \pgftransformxscale{-1}
                \fi
                \ifhvflip
                    \pgftransformxscale{-1}
                    \pgftransformyscale{-1}
                \fi
                \pgfpathmoveto{\pgfpoint
                    {-0.5*\overallwidth-5pt}
                    {0.5*\morphismheight}}
                \pgfpathlineto{\pgfpoint
                    {0.5*\overallwidth+5pt}
                    {0.5*\morphismheight}}
                \pgfpathlineto{\pgfpoint
                    {0.5*\overallwidth + \wedgewidth}
                    {-0.5*\morphismheight}}
                \pgfpathlineto{\pgfpoint
                    {-0.5*\overallwidth-5pt}
                    {-0.5*\morphismheight}}
                \pgfpathclose
                \pgfusepath{fill,stroke}
        \end{scope}
    \end{scope}
    }
}
\makeatother

\tikzstyle{every picture}=[baseline=-0.25em,scale=0.5]
\tikzstyle{label}=[font=\footnotesize,text height=1ex, text depth=0.15ex]

\tikzset{stateshape/.style={append after command={
   \pgfextra
        \draw[sharp corners, fill=none]%
    (\tikzlastnode.west)%
    [rounded corners=0pt] |- (\tikzlastnode.north)%
    [rounded corners=0pt] -| (\tikzlastnode.east)%
    [rounded corners=5pt] |- (\tikzlastnode.south)%
    [rounded corners=5pt] -| (\tikzlastnode.west);
   \endpgfextra}}}

\tikzset{effectshape/.style={append after command={
   \pgfextra
        \draw[sharp corners, fill=none]%
    (\tikzlastnode.west)%
    [rounded corners=0pt] |- (\tikzlastnode.south)%
    [rounded corners=0pt] -| (\tikzlastnode.east)%
    [rounded corners=5pt] |- (\tikzlastnode.north)%
    [rounded corners=5pt] -| (\tikzlastnode.west);
   \endpgfextra}}}
\tikzstyle{state}=[stateshape,inner sep=.4ex, node on layer=foreground, minimum width = 0.6cm, minimum height = 0.4cm, font = \small]
\tikzstyle{effect}=[effectshape,inner sep=.4ex, node on layer=foreground, minimum width = 0.6cm, minimum height = 0.4cm, font = \small]

\tikzstyle{statewide}=[state, minimum width = 1.2cm]
\tikzstyle{stateextrawide}=[state, minimum width = 1.4cm]
\usepackage{booktabs}
\usepackage{enumitem}
\setenumerate{label=\arabic*., ref=\arabic*}

\newtheorem{theorem}{Theorem}[section]
\newtheorem{proposition}[theorem]{Proposition}
\newtheorem{corollary}[theorem]{Corollary}
\newtheorem{lemma}[theorem]{Lemma}

\theoremstyle{definition}
\newtheorem{definition}[theorem]{Definition}
\newtheorem{example}[theorem]{Example}
\newtheorem{examples}[theorem]{Examples}
\newtheorem{remark}[theorem]{Remark}

\newtheorem{principle}{Principle}
\theoremstyle{remark}
\newtheorem*{remark*}{Remark}


\begin{document}

\title[A Categorical Reconstruction of Quantum Theory]{A Categorical Reconstruction of Quantum Theory}
\author[S.~Tull]{Sean Tull}	
\address{Department of Computer Science, University of Oxford}	
\email{sean.tull@cs.ox.ac.uk}

\begin{abstract}
We reconstruct finite-dimensional quantum theory from categorical principles.
That is, we provide properties ensuring that a given physical theory described by a dagger compact category in which one may `discard' objects is equivalent to a generalised finite-dimensional quantum theory over a suitable ring $S$. The principles used resemble those due to Chiribella, D'Ariano and Perinotti. Unlike previous reconstructions, our axioms and proof are fully categorical in nature, in particular not requiring tomography assumptions. Specialising the result to probabilistic theories we obtain either traditional quantum theory with $S$ being the complex numbers, or that over real Hilbert spaces with $S$ being the reals.
\end{abstract}

\maketitle

The Hilbert space formulation of quantum theory has remained difficult to interpret ever since it was first described by von Neumann~\cite{von1955mathematical}. Over the years this has led to numerous attempts to understand the quantum world from more basic, operational or logical principles. As well as being of physical interest, such new understandings are invaluable when developing logical frameworks for the modelling of quantum computation. Most recently there has been a great interest in `reconstructions' of quantum theory as a theory of information~\cite{Hardy2001QTFrom5,PhysRevA.84.012311InfoDerivQT,Hardy2011a,wilceRoyal,van2018reconstruction,selby2018reconstructing}. In these, quantum theory is singled out via operational axioms, referring to the likelihoods assigned to experimental procedures, from amongst all general probabilistic theories. 

Typically, a central aspect of any such theory is taken to be 
the specification of certain allowed physical \emph{systems} and \emph{processes} between them, which may be composed like pieces of apparatus in a laboratory. It is well-known that such `process theories' correspond precisely to \emph{monoidal categories}~\cite{coecke2011categories}, very general mathematical structures coming with an intuitive graphical calculus allowing one to reason using \emph{circuit diagrams}~\cite{selinger2011survey}. In the usual approach to reconstructions, one then explicitly adds further probabilistic structure using an  assumption known as \emph{finite tomography}, which enforces that the processes of any given type generate a finite-dimensional real vector space. Though physically motivated, such technical assumptions move these results further from formal logic.

 However, there is a second tradition in the literature characterised by avoiding these tomography assumptions, 
 and instead studying physical theories such as quantum theory purely in terms of their categorical or diagrammatic aspects. This programme, lying at the intersection of physics and computer science, is collectively referred to as 
 \emph{categorical quantum mechanics} (CQM)~\cite{abramskycoecke:categoricalsemantics}. The categorical approach is natural from a perspective in which processes are seen as fundamental to a physical theory, and provides an intuitive but formal logic for reasoning about quantum processes which is easily connected with other categorical formalisations of computation from across computer science. It has proven successful in studying numerous aspects of quantum foundations and quantum computation~\cite{CKbook,coecke2012strong,abramsky2009no,coecke2008interacting}.

  As such it is natural to ask whether one may recover quantum theory itself in this framework, providing us with a  categorical logic which fully axiomatises quantum processes. Indeed a reconstruction theorem for CQM has long been desired, with the need for such a result put forward by Coecke and Lal in~\cite{coecke2011categorical}. In this work, we present such a category-theoretic reconstruction of quantum theory.  

\para{The Principles}

Throughout we work in the framework of \emph{dagger theories}; symmetric monoidal categories coming with a \emph{dagger}, allowing us to `reverse' morphisms, and for which each object has a chosen \emph{effect} called \emph{discarding}, denoted $\discard{}$, which we can think of as the process of throwing away or ignoring a system. The theories we consider are also \emph{compact}, providing the ability to swap inputs and outputs of morphisms. Furthermore they meet the following axioms expressed using the graphical calculus. In the tradition of quantum reconstructions we refer to these as operational \emph{principles} our theory must satisfy. 

\begin{enumerate}[leftmargin=*]

\item \emph{Strong purification.} Any morphism $f$ has a dilation $g$ which is \emph{pure} in the sense that any dilation of $g$ is trivial. Diagrammatically: 
\[
\scalebox{0.85}{\input{figures/purif-simple1a.tikz}}

\quad
\text{ where  } 
\quad
\scalebox{0.85}{\input{figures/purif-simple2ia.tikz}}

\]
Furthermore we take pure morphisms to be compatible with composition and the dagger, making them an \emph{environment structure}~\cite{coecke2008axiomatic}, every suitable object to have a \emph{causal} pure state, and purifications to be \emph{essentially unique}~\cite{PhysRevA.84.012311InfoDerivQT}, meaning that any two of the same type are related by a suitable isomorphism. In our setting the above notion of purity in fact coincides with the usual one in terms of mixing (see Appendix~\ref{appendix:pure}).

\item \emph{Kernels.} Every morphism $f$ has a \emph{dagger kernel}, characterised by
\[
\scalebox{0.85}{\input{figures/kernel-def1.tikz}}
 \ 0 
\iff
(\exists ! h) \ 
\scalebox{0.85}{\input{figures/kernel-def2.tikz}}

\]
Kernels are closely related to axioms of~\cite{PhysRevA.84.012311InfoDerivQT} which allow one to associate a system to those states which are `perfectly distinguishable' from any given state. 

\item \emph{Pure exclusion.} Every causal pure state $\psi$ of a non-trivial system has $e \circ \psi = 0$ for some non-zero effect $e$.

\item \emph{Conditioning.}
For every pair of \emph{orthonormal} states $\ket{0}, \ket{1}$ and every pair of states $\rho, \sigma$ of some object there is a morphism $f$ with $f \circ \ket{0} = \rho$ and $f \circ \ket{1} = \sigma$. 
\end{enumerate}

Finally we also add basic compatibility requirements and make a mild assumption on the \emph{scalars} in our theory called \emph{boundedness}.

\para{Reconstruction}
We prove that any non-trivial dagger theory satisfying these principles is equivalent to a theory $\Quant{S}$ for a suitable ring $S$, a generalisation of the theory $\Quant{\mathbb{C}}$ of finite-dimensional complex Hilbert spaces and completely positive maps (Theorem~\ref{thm:main-reconstruction}). Moreover, we show that a further basic assumption on scalars provides $S$ with structure similar to the real or complex numbers (Lemma~\ref{lem:square roots}).

The statement and proof of these results are entirely category-theoretic, not assuming any form of tomography. When specialised to probabilistic theories, we obtain reconstructions of quantum theory $\Quant{\mathbb{C}}$ as well as the quantum theory $\Quant{\mathbb{R}}$ over real Hilbert spaces, and in fact we show that these are the unique probabilistic theories satisfying our principles. Recovering $\Quant{\mathbb{R}}$ is a pleasing consequence of our tomography-free approach, with most reconstructions ruling it out from the outset by assuming \emph{local tomography}; an example of a probabilistic reconstruction which does recover both theories is~\cite{hohn2017quantum}.

Our principles above were chosen to be as general as possible whilst allowing for this reconstruction result to hold. In future work it should be possible to arrive at more succinct reconstructions by deriving them from a smaller set of more natural, though potentially stronger, axioms.    

\para{The Proof}
More broadly, the heart of our proof lies in a very general approach to axiomatising quantum theories. This is based on a new description of superpositions from~\cite{superpos} using features called \emph{phased biproducts}. We use these to derive a general `recipe' for quantum reconstructions (Corollary~\ref{cor:recipe}) which we hope to be applicable to further such results in the future.

\para{Related work}
The above principles may be seen as a formulation of those of the Chiribella-D'Ariano Perinotti (CDP) reconstruction~\cite{PhysRevA.84.012311InfoDerivQT} for the setting of dagger compact categories.
\begin{center}
  \begingroup
  \def\arraystretch{1.2}
\begin{tabular}{@{}cc@{}}
\toprule
\textbf{CDP Axioms}             & \textbf{Categorical Features} \\ 
\hline
Coarse-graining                  & Conditioning                               \\ \hline
Causality                  & Discarding                                \\
\hline
Atomicity of composition   & \multirow{2}{*}{Environment structure}                 \\ \cline{1-1}
Purification               &                                                        \\ \hline
Perfect distinguishability &  Kernels, \\ \cline{1-1}
Ideal compressions         &  pure exclusion                                                      \\ \hline
\multicolumn{2}{c}{Essential uniqueness}                                  \\ \bottomrule
\end{tabular}
\endgroup
\smallskip
\end{center}
As such our result may be seen as an attempt to adapt CDP's result into one in the tradition of formal logic. Connections between categorical quantum mechanics and the CDP reconstruction were first considered in~\cite{coecke2011categorical}. Our work also draws on previous dagger-categorical treatments of (pure) quantum theory, including Jamie Vicary's characterisation of categories with complex numbers as scalars~\cite{vicary2011categorical}, and Chris Heunen's axiomatization of the category of Hilbert spaces~\cite{heunen2009embedding} which we discuss in detail in Appendix~\ref{apped:Hilbert}. Connections with other quantum reconstructions remain to be explored; notable are the recent probabilistic reconstructions due to Selby, Scandolo and Coecke~\cite{selby2018reconstructing} and van de Wetering~\cite{van2018reconstruction}.

\subsection*{Outline} 
The article is structured as follows. 
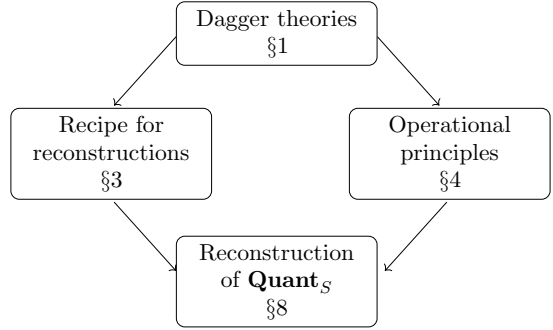
\begin{center}
\vspace{10pt}
\scalebox{0.8}{\input{figures/structure2.tikz}}

\vspace{10pt}
\end{center}

We begin by describing the general setting of dagger theories $\catC$, including the quantum theories $\Quant{S}$, in Section~\ref{sec:setup}. 

In Sections~\ref{sec:building} and~\ref{sec:recipe} we give our general recipe for quantum reconstructions, based on describing superpositions in the category $\catB:=\catC_\pure$. 

Those only interested in our main reconstruction can skip to Section~\ref{sec:axioms} where we introduce the operational principles for a dagger theory $\catC$. 

The next sections contain most of the proof of our main results. Section~\ref{sec:consequences} studies the consequences of our principles and Section~\ref{sec:derivingsups} derives the presence of superpositions from them, allowing us to apply our recipe. In Section~\ref{sec:PureProcessProperties} we find further properties of our theory by axiomatizing its extended category of pure morphisms $\catA := \plusI{\catB}$.

The main reconstruction results are deduced in Section~\ref{sec:reconstructions-listed}, where we return to the dagger theory $\catC$. We then quickly specialise these to probabilistic theories in Section~\ref{sec:probTheories}. We close by discussing our results and open questions, with details expanded on in the appendices. 

\section{Setup} \label{sec:setup}

In the process-theoretic approach, a theory of physics specifies certain \emph{systems} and composable \emph{processes} between them. In other words, a theory corresponds to a \emph{category} $\catC$ whose objects $A, B, C \dots$ are systems and whose morphisms $f \colon A \to B$ are processes. It is natural to suppose that we may also place systems and processes `side-by-side' via some operation $A, B \mapsto A \otimes B$ and $f, g \mapsto f \otimes g$, and it is well-known that this is captured by the notion of a \emph{(symmetric) monoidal category} $(\catC, \otimes, I)$. For an overview of symmetric monoidal categories in physics see~\cite{coecke2011categories}.  

Along with the notation $f \colon A \to B$ for morphisms in a category, we will make use of the \emph{graphical calculus} for monoidal categories~\cite{selinger2011survey}
 in which morphisms $f \colon A \to B$ are depicted 
$\setlength\morphismheight{3mm}\begin{pic}
  \node[morphismflip,font=\small] (f) at (0,0) {$f$};
  \draw (f.south) to +(0,-.15)node[below] {$A$};
  \draw (f.north) to +(0,.15) node[above] {$B$};
\end{pic}$, with identity morphisms as vertical wires and composition by 
\[ 
\scalebox{0.85}{\input{figures/composition-morphism.tikz}}

\qquad  \qquad
\scalebox{0.85}{\input{figures/tensor-morphism.tikz}}

\]
Any monoidal category comes with a \emph{unit object} $I$, and we call morphisms $\rho \colon I \to A$, $e \colon A \to I$ and $s \colon I \to I$ \emph{states, effects} and \emph{scalars} respectively. In diagrams (the identity on) $I$ is an empty picture so that these are depicted as:
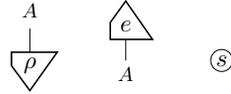
\[
\scalebox{0.85}{\input{figures/state-effect-scalar.tikz}}

\]
The scalars form a commutative monoid under composition and come with a multiplication $f \mapsto s \cdot f$ on morphisms given graphically by drawing $s$ alongside $f$. 

As well as being symmetric monoidal, we will be interested in categories with extra structure. Recall that a \emph{dagger} monoidal category is one coming with a functor $(-)^{\dag} \colon \catC^{\op} \to \catC$ satisfying $A^{\dagger} = A$ and $(f^{\dagger})^ \dagger = f$ for all objects $A$ and morphisms $f$, and being compatible with $\otimes$ in a suitable sense. Graphically the dagger is represented by turning pictures upside-down: 
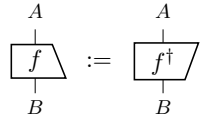
\[
\scalebox{0.85}{\input{figures/dagger-i.tikz}}

\]
Borrowing terminology from Hilbert spaces, in a dagger category an \emph{isometry} is a morphism $f \colon A \to B$ with $f^{\dagger} \circ f = \id{A}$, and a \emph{unitary} additionally has $f \circ f^{\dagger} = \id{B}$. A morphism $f$ is \emph{positive} when $f = g^{\dagger} \circ g$ for some morphism $g$.

A further manner in which we may relax our requirements on diagrams is to be free to exchange inputs and outputs of morphisms. A \emph{dagger compact category}~\cite{Selinger2007139} is a symmetric monoidal category with a suitably compatible dagger, and for which for every object $A$ there exists another object $A^*$ and a state $\tinycup \colon I \to A^* \otimes A$ satisfying the \emph{snake equations}:
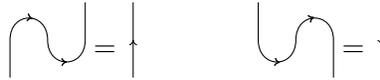
\[
  \begin{pic}[scale=.5]
    \draw[arrow=.36, arrow=.66] (0,0) to (0,1) to[out=90,in=90,looseness=2] (1,1) to[out=-90,in=-90,looseness=2] (2,1) to (2,2);
  \end{pic}
  =
  \begin{pic}[scale=.5]
  \draw[arrow=.5] (0,0) to (0,2);
  \end{pic}
  \qquad \qquad
  \begin{pic}[scale=.5]
    \draw[reverse arrow=.37, reverse arrow=.67] (0,0) to (0,1) to[out=90,in=90,looseness=2] (-1,1) to[out=-90,in=-90,looseness=2] (-2,1) to (-2,2);
  \end{pic}
  =
  \begin{pic}[scale=.5]
  \draw[reverse arrow=.5] (0,0) to (0,2);
  \end{pic}
\]
where $\tinycap = \tinycup^{\dagger} \circ \sigma$ for the `swap' isomorphism $\sigma$ which exists in any symmetric monoidal category. In particular the identity morphism on the objects $A$ and $A^*$ are depicted with upward and downward arrows respectively. For any morphism $f \colon A \to B$ we may then define morphisms 
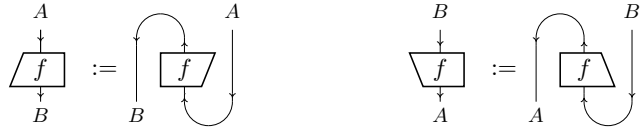
\[
\scalebox{0.85}{\input{figures/dual-morphisms.tikz}}

\]
Now we will be primarily interested in categories with an interpretation as operational procedures one may perform in some domain of physics. These are distinguished by an extra feature, common to all causal operational theories~\cite{PhysRevA.84.012311InfoDerivQT}; the ability to `discard' or `ignore' systems.

\begin{definition} \label{def:dagger-theory}
A \emph{monoidal category with discarding} is a monoidal category $\catC$ with a chosen effect $\discard{A}$ on each object $A$ such that 
\[
\scalebox{0.85}{\input{figures/discard_axioms.tikz}}

\]
By a \emph{dagger theory} $(\catC, \discard{})$ we then mean a dagger category with discarding. A dagger theory is \emph{compact} when $\catC$ is dagger compact and discarding satisfies
\[
\scalebox{0.85}{\input{figures/discard-compact.tikz}}

\]
\end{definition}

The presence of discarding morphisms in fact allows for a general treatment of causality~\cite{coecke2013causal}, and we call a morphism $f \colon A \to B$ \emph{causal} when $\discard{B} \circ f = \discard{A}$. 

Finally, all of our categories of interest will come with \emph{zero morphisms}, meaning that there is a morphism $0 \colon A \to B$ between any two objects, together satisfying
\[
0 \circ f = 0  \qquad f \circ 0 = 0
\]
and $0 \otimes f = 0$ in the monoidal case, for all morphisms $f$. A \emph{zero object} $0$ is then one with $\id{0} = 0$, or equivalently that is both initial and terminal.

Let us now meet our main examples of dagger categories and theories. 

\begin{example} \label{ex:MatS}
Let $S$ be a semi-ring (i.e.~a ring without a notion of subtraction) which is commutative and has an \emph{involution}, i.e.~an automorphism $s \mapsto s^{\dagger}$ with $s^{\dagger \dagger} = s$ for all $s \in S$. The dagger compact category $\Mat_{S}$ has as objects the natural numbers, with morphisms $M \colon n \to m$ being $n \times m$ matrices over $S$. Composition is that of matrices, with $(M^{\dagger})_{i,j} = M_{j,i}^{\dagger}$ and $M \otimes N$ given by the Kronecker product.
\end{example}

\begin{example} \label{ex:classical}
We can model classical probabilistic physics on finite systems by the compact dagger theory $\Class := \Mat_{\mathbb{R^+}}$, where morphisms are `unnormalised stochastic matrices'. Here discarding is given by $\discard{n} = (1, \dots 1) \colon n \to 1$. 

 Classical `possibilistic' physics is instead described by the compact dagger theory $\Rel$ whose objects are sets and morphisms $A \to B$ are relations $R \subseteq A \times B$. Here $\otimes$ is Cartesian product, the dagger is relational converse, and $\discard{A}$ is the unique function from $A$ into a singleton set $I=\{\ast\}$.
\end{example}

\begin{example}
Quantum theory may be described in terms of any of three dagger compact categories whose objects are finite-dimensional complex Hilbert spaces. 
\begin{center}
  \begingroup
  \def\arraystretch{1.1}
\vspace{1pt}
\hspace{-3pt}
\begin{tabular}{@{}lll@{}}
  \toprule
& Morphisms $\hilbH \to \hilbK$                           & Scalars                      \\\midrule
$\FHilb$  \hspace{-5pt} & \hspace{-3pt} linear maps $\hilbH \to \hilbK$                         & $\mathbb{C}$ \\ 
$\FHilbP$ \hspace{-5pt} & \hspace{-3pt} linear maps $\hilbH \to \hilbK$ up to global phase      & $\mathbb{R^+}$  \\ 
$\Quant{}$ \hspace{-5pt}
  & completely positive \hspace{-3pt} linear maps $B(\hilbH) \to B(\hilbK)$ \hspace{-7pt} & $\mathbb{R^+}$  \\ 
  \bottomrule
\end{tabular}
  \endgroup
\vspace{1pt}
\end{center}

Most simply, in the category $\FHilb$ morphisms are linear maps $f \colon \hilbH \to \hilbK$. Physically, however, any two such maps $f ,g$ are in fact indistinguishable whenever they are equal up to \emph{global phase} i.e.~$f = e^{i \theta} \cdot g$ for some $\theta \in [0,2\pi)$.
Taking equivalence classes $[f]$ of linear maps $f$ under this relation yields the category $\FHilbP$ of `pure' quantum processes. Extending this to consider arbitrary mixed processes between quantum systems leads to the category $\Quant{}$ where morphisms $\hilbH \to \hilbK$ are completely positive linear maps $f \colon B(\hilbH) \to B(\hilbK)$. 

In each case $\otimes$ is the tensor product of Hilbert spaces, $I= \mathbb{C}$, and $f^{\dagger}$ is the adjoint of the linear map $f$. In particular, states and effects on $\hilbH$ in $\FHilb$ correspond to vectors $\psi \in \hilbH$ while in $\Quant{}$ they are unnormalised density matrices $\rho \in B(\hilbH)$. Scalars in $\FHilb$ may be equated with elements of $\mathbb{C}$, while in $\FHilbP$ and $\Quant{}$ they are `unnormalised probabilities' $\in \mathbb{R^+}$. 

$\Quant{}$ is moreover a dagger theory, with $\discard{\hilbH} \colon B(\hilbH) \to \mathbb{C}$ being the map $a \mapsto \Tr(a)$. In particular a morphism here is causal precisely when it is trace-preserving as a completely positive map. 
\end{example}

\begin{example}
More broadly one may also consider the dagger symmetric monoidal category $\Hilb$ of (arbitrary) Hilbert spaces and bounded linear maps. However it is not dagger compact, and neither is its quotient $\HilbP$ under global phases.
\end{example}

\noindent
There is a well-known description of the dagger theory $\Quant{}$ in terms of its subcategory $\FHilbP$ due to Selinger~\cite{Selinger2007139} which allows it to be generalised. 

\begin{definition}
Let $\catB$ be a dagger compact category. The category $\CPM(\catB)$ is defined as having the same objects as $\catB$, with morphisms $A \to B$ being morphisms in $\catB$ of the form
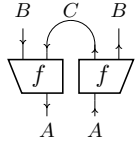
\[
\scalebox{0.85}{\input{figures/CPM-map.tikz}}

\]
for some $f \colon A \to C \otimes B$. Then $\CPM(\catB)$ forms a compact dagger theory, with
\[
\scalebox{0.85}{\input{figures/CPM-discarding.tikz}}

\]
\end{definition}
\noindent The motivating example is as follows. 

\begin{example}
$\Quant{} \simeq \CPM(\FHilb) \simeq \CPM(\FHilbP)$.
\end{example}
\noindent
Noting that $\FHilb \simeq \Mat_{\mathbb{C}}$ suggests the following generalisation.

\begin{definition} \label{def:QuantS}
For any commutative involutive semi-ring $S$ we define a dagger theory 
\[
\Quant{S} := \CPM(\Mat_S)
\]
Explicitly, morphisms $n \to m$ are $S$-valued matrices of the form ${\sum^k_{i=1} (M_i)_* \otimes M_i}$, where each $M_i$ is an $n \times m$ matrix over $S$. 
\end{definition}

\begin{example} 
Standard quantum theory is $\Quant{} \simeq \Quant{\mathbb{C}}$, with the above being the Kraus decomposition of a completely positive map. Another physically interesting example is $\Quant{\mathbb{R}}$ which describes quantum theory \emph{on real Hilbert spaces}~\cite{stueckelberg1960quantum,Hardy2012Holism}. For more on generalised quantum theories see~\cite{fantastic}. 
\end{example}

Our main goal in this article will be to find conditions ensuring that a dagger theory is equivalent to $\Quant{}$, or more generally $\Quant{S}$ for some semi-ring $S$. 

\section{Superpositions} \label{sec:building}

Key to our understanding of theories such as $\Quant{}$ will be answering the following basic question: given Hilbert spaces $\hilbH$ and $\hilbK$, how do we describe the space $\hilbH \oplus \hilbK$, modelling superpositions from each, categorically? In fact a well-known notion seemingly answers this question.


In a dagger category we call morphisms $f \colon A \to B$ and $g \colon C \to B$ \emph{orthogonal} when $g^{\dagger} \circ f = 0$. 

\begin{definition} \label{def:biproducts}~\cite{selinger2008idempotents}
In a dagger category $\catB$ with a zero morphisms, a \emph{dagger biproduct} of objects $A, B$ is another object $A \biprod B$ together with orthogonal isometries $\begin{tikzcd}
A \rar{\coproj_A} & A \biprod B & B \lar[swap]{\coproj_B}
\end{tikzcd}$
such that for all $f \colon A \to C$ and $g \colon B \to C$ there exists a unique morphism $h$ making the following commute:
\begin{equation} \label{eq:biprod-eq}
\begin{tikzcd}
A \rar{\coproj_A} \drar[swap]{f} & A \biprod B
\dar[dashed]{h}
& B \lar[swap]{\coproj_B} \dlar{g} \\ 
& C &
\end{tikzcd} 
\end{equation}
We denote the unique such morphism by $h = [f,g]$. 
\end{definition}

\noindent
In categorical language, the condition above states that $A \biprod B$ is a coproduct of $A, B$, and the dagger makes it also their product compatibly. Finite dagger biproducts $A_1 \biprod \dots \biprod A_n$ may be defined similarly, and in fact exist precisely when $\catB$ has binary biproducts and a zero object.

\begin{example}
Each category $\Mat_{S}$ has dagger biproducts given on objects by $n \biprod m = n + m$. In particular, $\FHilb$ has dagger biproducts given by direct sum $\hilbH \biprod \hilbK$ of Hilbert spaces, as does $\Hilb$.
\end{example}

Biproducts are a general way to describe addition and matrix-like features internally to a category, providing a well-behaved addition operation on morphisms $f, g \colon A \to B$ by
\begin{equation*}
f + g := (
\begin{tikzcd}
A \rar{\Delta} & A \biprod A \rar{[f,g]} & B
\end{tikzcd}
)
\end{equation*}
where $\Delta^{\dagger} = [\id{A}, \id{A}]$. In any monoidal category with finite biproducts which are `distributive' (as in any dagger compact category) this makes the scalars $S = \catC(I,I)$ a commutative semi-ring, and there is a full embedding $\MatS \hookrightarrow \catC$ sending each object $n$ to the object
\[
 n \cdot I:= \overbrace{I \biprod \dots \biprod I}^n
\]
In $\Hilb$ these additive features describe quantum superpositions. However, there is a problem: while $\Hilb$ has biproducts, $\HilbP$ does not, and only the morphisms of the latter category have a direct physical interpretation. Nonetheless, $\HilbP$ has a similar feature describing superpositions, introduced in~\cite{superpos}.

\begin{definition} \label{def:ph_coprod}
A \emph{phased dagger biproduct} of $A, B$ is an object $A \pbiprod B$ with orthogonal isometries
$\begin{tikzcd}
A \rar{\coproj_A} & A \pbiprod B & B \lar[swap]{\coproj_B}
\end{tikzcd}$
 such that:
\begin{enumerate}
\item For all 
$f \colon A \to C$ and $g \colon B \to C$ there exists some morphism $h$ making~\eqref{eq:biprod-eq} commute (writing $\pbiprod$ in place of $\biprod$);
\item 
Any pair of such morphisms $h, h'$ have $h' = h \circ U$ for some endomorphism $U$ which is a \emph{phase}, meaning that $U \circ \pcoproj_A = \pcoproj_A$ and $U \circ \pcoproj_B = \pcoproj_B$;
\[
\begin{tikzcd}
 \arrow[loop left, "U",distance = 2em] A \pbiprod B \rar[shift right = 2, swap]{h'} \rar[shift left = 2]{h}   & D
\end{tikzcd}
\] 
\item 
Whenever $U$ is a phase then so is $U^\dagger$.
\end{enumerate}
\end{definition}

Again a zero object along with these suffices to describe finite phased dagger biproducts $A_1 \pbiprod \dots \pbiprod A_n$. A biproduct is then simply a phased biproduct whose only phase is the identity. 

\begin{example}
$\HilbP$ and $\FHilbP$ have phased dagger biproducts given by the biproducts in $\Hilb$, i.e.~the direct sum of Hilbert spaces. More precisely, for any such direct sum
$\begin{tikzcd}
\hilbH \rar{\coproj_\hilbH} & \hilbH \biprod \hilbK & \hilbK \lar[swap]{\coproj_\hilbK}
\end{tikzcd}$
 it is easy to see the induced morphisms $[\coproj_\hilbH]$ and $[\coproj_\hilbK]$ in $\HilbP$ form a phased dagger biproduct with all phases of the form $[U]$ where $U$ is a unitary on $\hilbH \oplus \hilbK$ in $\Hilb$ of the form 
\[
U = 
  \begin{pmatrix}
    e^{i \theta} \cdot \id{\hilbH} & 0 \\
    0 & \id{\hilbK} 
  \end{pmatrix}
\]
for some $\theta \in [0,2\pi)$.
\end{example}

\begin{example}
More generally, by a choice of \emph{global phases} on a dagger symmetric monoidal category $\catA$ we mean a subgroup $\mathbb{P}$ of its unitary scalars. We write $\catA_{\sim}$ for its category of equivalence classes under $f \sim g$ whenever $f = u \cdot g$ for some $u \in \mathbb{P}$. Then when $\catA$ has distributive dagger biproducts, these form phased dagger biproducts in $\catA_{\sim}$.
\end{example}

\noindent
In fact there is a general way of reversing this example for phased biproducts with some extra  properties. Call a morphism $f \colon A \pbiprod B \to C \pbiprod D$ \emph{diagonal} when $f \circ \pcoproj_A = \pcoproj_C \circ f_A$ and $f \circ \pcoproj_B = \pcoproj_D \circ f_B$ for some $f_A, f_B$. 

In general we will call a dagger monoidal category \emph{state-inhabited} when every non-zero object $A$ has at least one isometry $I \to A$. 
 
\begin{definition} \label{def:sup_properties}
A dagger compact category $\catB$ has the \emph{superposition properties} when it is state-inhabited and has finite phased dagger biproducts such that for all phases $U$ and positive diagonal morphisms $p, q$ on $A \pbiprod B$ we have 
\begin{equation} \label{eq:positive-cancellation}
p = q \circ U \implies p = q
\end{equation}
\end{definition}

By an \emph{embedding} of dagger (symmetric) monoidal categories we mean a faithful dagger (symmetric) monoidal functor $F \colon \catC \to \catD$. It is an \emph{equivalence} $\catC \simeq \catD$ when it is full and every object of $\catD$ is unitarily isomorphic to $F(A)$ for some $A \in \catC$.

\begin{theorem} \label{thm:onPHConstr}\cite{superpos}
Let $\catB$ be a dagger compact category with the superposition properties. Then there is a dagger compact category\footnote{Here we always use the $\GP$ construction for dagger categories, which in \cite{superpos} is denoted $\plusIdagsym$.} $\plusI{\catB}$ with finite dagger biproducts and a choice of global phases such that 
\[
\catB \ \simeq \ \plusI{\catB}_{\sim}
\]
Further, positive morphisms $p, q \in \plusI{\catB}$ satisfy $p \sim q \implies p = q$.
\end{theorem}

We write $[-] \colon \plusI{\catB} \to \catB$ for the dagger monoidal functor which takes equivalence classes under $\sim$.

\begin{example}
$\FHilbP$ has the superposition properties and we have 
\[
\plusI{\FHilbP} \simeq \FHilb
\]
with global phases $ \{ z \in \mathbb{C} \mid |z| = 1\}$.
\end{example}

\begin{remark*}
We will not require the explicit description of the category $\plusI{\catB}$ here, but it is constructed as follows~\cite{superpos};
objects are phased biproducts $\obb{A} = A \pbiprod I$ and morphisms $f \colon \obb{A} \to \obb{B}$ are diagonal morphisms with $f \circ \pcoproj_I = \pcoproj_I$.

In fact the $\GP(\catB)$ construction extends more generally to (dagger symmetric) monoidal categories with phased biproducts which are `distributive' and satisfy a weakening of~\eqref{eq:positive-cancellation}; see\cite{superpos}. In particular we similarly have
\[
\plusI{\HilbP} \simeq \Hilb 
\]
\end{remark*}

\section{A Recipe for Reconstructions} \label{sec:recipe}

Our description of superpositions provides a general result for reconstructing quantum-like theories. For this we now need to know when, like $\Quant{S}$, a dagger theory is of the form $\CPM(\catB)$.

In fact those theories arising from the $\CPM$ construction have been axiomatized by Coecke~\cite{coecke2008axiomatic}. Firstly, in any category with discarding, a morphism $g \colon A \to B \otimes C$ is said to be a \emph{dilation} of a morphism $f \colon A \to B$ when 
\[
\scalebox{0.85}{\input{figures/dilation.tikz}}

\]
Now let $\catC$ be a compact dagger theory. An \emph{environment structure}~\cite{coecke2010environment} on $\catC$ is a choice of dagger compact subcategory $\catC_{\prepure}$ such that every morphism of $\catC$ has a dilation in $\catC_{\prepure}$, and additionally all morphisms $f \colon A \to B$, $g \colon A \to C$ in $\catC_{\prepure}$ satisfy the \emph{CP axiom}:
\begin{equation} \label{eq:CP axiom}
\scalebox{0.85}{\input{figures/CP.tikz}}

\end{equation}

\begin{example}
The dagger theory $\Quant{}$ has an environment structure given by the subcategory $\FHilbP$. The ability to dilate processes in this way is referred to as \emph{Stinespring dilation}~\cite{stinespring1955positive} or `purification'~\cite{PhysRevA.84.012311InfoDerivQT}.
\end{example}

If $\catB$ is a dagger compact category then $\CPM(\catB)$ has an environment structure with $\CPM(\catB)_{\prepure}$ being the image $\Dbl{\catB}$ of the dagger monoidal functor $\Dbl{(-)}  \colon \catB \to \CPM(\catB)$ which preserves objects and is given on morphisms by 
\[
\scalebox{0.85}{\input{figures/DBL-2.tikz}}

\]
\noindent
Conversely, by an \emph{embedding} or \emph{equivalence} of dagger theories $F \colon \catC \to \catD$ we mean one of dagger symmetric monoidal categories which also preserves discarding in that $\discard{F(I)} = \phi \circ F(\discard{I})$ holds for its given isomorphism $\phi \colon F(I) \simeq I$, and which preserves zero morphisms when they exist. Then for any compact dagger theory $\catC$ with an environment structure $\catC_\prepure$ there is an equivalence of dagger theories $\CPM(\catC_{\prepure})\simeq \catC$ given by
\[
\scalebox{0.85}{\input{figures/CPM-map-smaller.tikz}}

\mapsto
\scalebox{0.85}{\input{figures/env_struc_map.tikz}}

\]

\begin{lemma} \label{lem:CPMsCoincide}
Let $\catB$ be a dagger compact category with the superposition properties. Then the functor $[-] \colon \plusI{\catB} \to \catB$ extends to an equivalence of compact dagger theories 
\[
 \CPM(\plusI{\catB})
 \simeq
\CPM(\catB)
\]
\end{lemma}
\begin{proof}
Since $[-]$ is a wide full dagger symmetric monoidal functor and is surjective on objects up to unitary it lifts to such a functor $\CPM(\plusI{\catB}) \to \CPM(\catB)$. For faithfulness we require that
\[
\scalebox{0.85}{\input{figures/CPM-mapsmall.tikz}}

\sim
\scalebox{0.85}{\input{figures/CPM-mapsmall2.tikz}}

\implies
\scalebox{0.85}{\input{figures/CPM-mapsmall.tikz}}

=
\scalebox{0.85}{\input{figures/CPM-mapsmall2.tikz}}

\]
which after bending wires is just the last line of Theorem~\ref{thm:onPHConstr}.
\end{proof}

\noindent
For a monoid $S$ with involution $(-)^\dagger$ we write $S^{\pos}$ for the collection of elements which are positive, i.e.~of the form $s^{\dagger} \cdot s$ for some $s \in S$. 

\begin{corollary} \label{cor:recipe}
Let $\catC$ be a compact dagger theory with an environment structure $\catC_{\prepure}$ which has the superposition properties. Then there is an embedding of dagger theories 
\begin{equation*} 
\Quant{S} \hookrightarrow \catC
\end{equation*}
for some commutative involutive semi-ring $S$ with $\catC_{\prepure}(I,I) \simeq S^{\pos}$ as monoids.
\end{corollary}
\begin{proof}
Since $\plusI{\catC_{\prepure}}$ is dagger compact with dagger biproducts, we've seen that its scalars $S$ form a commutative involutive semi-ring, giving an embedding $\Mat_{S} \hookrightarrow \plusI{\catC_{\prepure}}$ and hence embeddings
\[
\begin{tikzcd}
\CPM(\Mat_{S})
\rar[hook,swap]{}
&
\CPM(\plusI{\catC_{\prepure}}) 
\rar{\sim}[swap]{\ref{lem:CPMsCoincide}}
&
\CPM(\catC_{\prepure})
\rar{\sim}
&
\catC
\end{tikzcd}
\]
Finally, let $R=\catC_{\prepure}(I,I)$. By the CP axiom and the last line of Theorem~\ref{thm:onPHConstr} the map $R \to R^{\pos}$ sending $r \mapsto r^\dagger \circ r$ and the map $[-] \colon S^{\pos} \to R^{\pos}$ are both monoid isomorphisms. 
\end{proof}

To obtain a full reconstruction, it remains to find further conditions making the embedding an equivalence. 

\begin{example}
Let $S$ be a commutative involutive semi-ring in which every non-zero element is invertible and for all positive elements $p$ we have $p^2 = 1 \implies p =1$. Then the above environment structure on $\Quant{S}$ has the superposition properties and is state-inhabited. Examples of such $S$ include $\mathbb{C}, \mathbb{R}, \mathbb{R^+}$ and the Booleans $\mathbb{B}$.
\end{example}

\section{The Operational Principles} \label{sec:axioms}

For now let us leave aside superpositions and our recipe for reconstructions, and return to the general setting of dagger theories outlined in Section \ref{sec:setup}. We now present axioms for theories of the form $\Quant{S}$ with a clearer operational meaning than the superposition properties, motivated by the CDP reconstruction for probabilistic theories~\cite{PhysRevA.84.012311InfoDerivQT}. We will consider dagger theories $(\catC,\discard{})$ coming with zero morphisms and the following features. 

\subsection{Purification}
Firstly, crucial to~\cite{PhysRevA.84.012311InfoDerivQT} is the ability in quantum theory to dilate any process to one of the following form. 

\begin{definition} \label{def:pure}
In any monoidal category with discarding and zero morphisms, we call a morphism $f \colon A \to B$ \emph{pure} when either $f = 0$ or $f$ satisfies 
\begin{equation} \label{eq:whole}
\scalebox{0.85}{\input{figures/whole-i.tikz}}

\end{equation}
for some causal state $\rho$, for all such dilations $g$.
\end{definition}

This notion of purity is due to Chiribella~\cite{chiribella2014distinguishability}. In the probabilistic setting, such as in~\cite{chiribella2010purification}, purity is more typically defined in terms of mixtures of processes. Recently there have also been several proposed alternative diagrammatic definitions of purity~\cite{selby2017leaks,cunningham2017purity}. However, for theories satisfying our principles all of these notions coincide, as we show in Appendix~\ref{appendix:pure}. Like~\cite{PhysRevA.84.012311InfoDerivQT} we consider theories in which morphisms have well-behaved pure dilations, as follows. 

\begin{principle}[Strong purification]
The collection $\catC_\pure$ of pure morphisms forms an environment structure on $\catC$. Moreover, every non-zero object of $\catC$ has a causal pure state, and for all pure morphisms $f, g$ we have 
\begin{equation}
\scalebox{0.85}{\input{figures/EU-purif.tikz}}
 \label{eq:essential-uniqueness}
\end{equation}
for some pure unitary $U \colon C \to C$. We refer to any pure dilation $g$ of a process $f$ a \emph{purification} of $f$, and condition \eqref{eq:essential-uniqueness} as stating that these are \emph{essentially unique}.
\end{principle}

In particular the above tells us that pure morphisms are closed under composition and contain all identity morphisms, giving the rule
\begin{equation} \label{eq:noleaks}
\scalebox{0.85}{\input{figures/noleaks.tikz}}

\end{equation}
for some causal state $\rho$, for all non-zero morphisms $f$. This is referred to as having `no leaks' in~\cite{selby2017leaks}. In particular it follows that all unitaries in $\catC$ are pure.

\subsection{Kernels}
Our next principle draws on other axioms of~\cite{PhysRevA.84.012311InfoDerivQT} encoding the following feature of both quantum and classical theory. 

Given a (causal) state $\rho$ of some system we can consider the other states $\sigma$ which may be `perfectly distinguished' from $\rho$ by some experimental test; in fact this holds iff~$\rho^{\dagger} \circ \sigma = 0$. Crucially, the collection of such states itself forms a system; indeed in the quantum case $\sigma$ is of this form iff it factors over the inclusion $\supp(\rho)^{\bot} \hookrightarrow \hilbH$. An existing categorical notion extends these ideas to arbitrary morphisms. 

\begin{definition}
Let $\catC$ be a dagger category with zero morphisms. A \emph{dagger kernel} of $f \colon A \to B$ is an isometry ${\ker(f) \colon \Ker(f) \to A}$ with ${f \circ \ker(f) = 0}$, such that every $g \colon C \to A$ with $f \circ g = 0$ has $g = k \circ h$ for some (necessarily unique) morphism $h$.
\[
\begin{tikzcd}
\Ker(f) \rar{\ker(f)} & A \rar[shift left=2.5]{f} \rar[shift right=2.5,swap]{0} & B \\ 
 C \arrow[ur,swap, "g"] \arrow[u, dashed, "\exists ! h"] 
\end{tikzcd}
\]
We say $\catC$ \emph{has dagger kernels} when every morphism has a dagger kernel. 
Dually, a \emph{dagger cokernel} of $f$ is a morphism $\coker(f) \colon B \to \Coker(f)$ which any morphism with $g \circ f = 0$ factors over, and for which $\coker(f)^{\dagger}$ is an isometry.
\end{definition}

Dagger kernels were first studied in detail by Heunen and Jacobs in~\cite{heunen2010quantum}. Whenever we use the term `(co)kernel' we will always mean `dagger (co)kernel'. 

Any two kernels of the same morphism are equal up to a unique unitary on their domain, and so we may identify them and speak of `the' kernel of a morphism. A category with dagger kernels also has cokernels via $\coker(f) = \ker(f^{\dagger})^{\dagger}$ and a zero object given by $0 = \Ker(\id{A})$ for any object $A$. Any kernel $k$ comes with a \emph{complement} defined by  
\[
k^{\bot} := \coker(k)^{\dagger}
\]
In fact this makes the collection of kernels on any fixed object into an \emph{orthomodular lattice}, just like the subspaces of a Hilbert space~\cite{heunen2010quantum}. 

We will need an extra assumption about kernels and their complements. Let us call a pair of effects $d, e \colon A \to I$ \emph{total} when for all morphisms $f, g \colon B \to A$ whenever $d \circ f = d \circ g$ and $e \circ f = e \circ g$ we have $\discard{A} \circ f = \discard{A} \circ g$. 

\begin{principle}[Kernels]
The category $\catC$ has dagger kernels, which moreover are \emph{causally complemented} in that for all dagger kernels $k \colon K \to A$ the effects
\begin{equation} \label{eq:caus-comp-pair}
\scalebox{0.85}{\input{figures/causal-comp-pair1.tikz}}

\quad
\text{ and }
\ \ 
\quad
\scalebox{0.85}{\input{figures/causal-comp-pair2.tikz}}

\end{equation}
are total.
\end{principle}

Causal complementation is a natural requirement if we are to think of the above effects as the two distinct outcomes of some binary test we may perform on the system $A$. 

\begin{example}
$\Quant{}$ has dagger kernels which are causally complemented. Since any morphism satisfies $\discard{} \circ f = 0 \implies f = 0$, it suffices to consider kernels of effects $e$. For a state $\rho \in B(\hilbH)$, $\ker(\rho^{\dagger})$ is precisely (the CP map determined by) the inclusion $\supp(\rho)^{\bot} \hookrightarrow \hilbH$, with complement given by $\supp(\rho) \hookrightarrow \hilbH$. 

 A state $\sigma$ in fact factors over the latter morphism  precisely when it appears in some convex decomposition of $\rho$, making it an \emph{ideal compression scheme} for $\rho$ in the sense of~\cite{PhysRevA.84.012311InfoDerivQT}.
These kernels are all pure and restrict to give dagger kernels in $\FHilb$ and $\FHilbP$. More broadly $\Hilb$ and $\HilbP$ have kernels, being the motivating example in~\cite{heunen2010quantum}.
\end{example} 

\begin{example} \label{ex:dag-ker-in-ClassProb}
Dagger kernels also exist in our classical theories $\Class$ and $\Rel$. In $\Class$, $\ker(M)$ is the usual kernel of a matrix $M$. In $\Rel$ a morphism $R \colon A \to B$ has $\ker(R) = \{a \in A \mid \nexists b \in B \ \  R(a,b)\}$.
\end{example}

\subsection{Pure exclusion}

In order for pure processes to behave with respect to zero morphisms as they do in $\Quant{}$ we require an extra principle. Let us call an object $A$ \emph{trivial} when $\discard{A} \colon A \to I$ is a unitary or $A$ is a zero object, and a theory \emph{trivial} when every object is trivial.

\begin{principle}[Pure Exclusion]
For every causal pure state $\psi$ of a non-trivial object $A$ there is a non-zero effect $e$ with
\[
\scalebox{0.85}{\begin{tikzpicture}
	\begin{pgfonlayer}{nodelayer}
		\node [style=kpointadj] (0) at (0, 1) {$e$};
		\node [style=kpoint] (1) at (0, -0.75) {$\psi$};
		\node [style=none] (2) at (1.5, -0) {$=$};
	\end{pgfonlayer}
	\begin{pgfonlayer}{edgelayer}
		\draw [style=none] (1) to (0);
	\end{pgfonlayer}
\end{tikzpicture}}
 \ 0
\]
\end{principle}

Pure exclusion encodes the idea that any pure state of a non-trivial system may be ruled out (excluded) by some experimental test (following from `perfect distinguishability' in \cite{PhysRevA.84.012311InfoDerivQT}). We will meet some natural equivalent conditions to pure exclusion shortly, in Section~\ref{sec:consequences}.

\subsection{Conditioning}

Finally, we will need to consider one very mild extra property, which allows us to form processes whose output is `conditional' on their input, in the following sense. We call a pair of states $\ket{0}, \ket{1}$ of the same object \emph{orthonormal} when they are orthogonal isometries, i.e.~$ \ket{0}^\dagger \circ \ket{0} = \id{I} = \ket{1}^\dagger \circ \ket{1}$ and $\ket{1}^\dagger \circ \ket{0} = 0$.

\begin{principle}[Conditioning]
For every pair of orthonormal states $\ket{0}, \ket{1}$ of an object $A$ and any states $\rho, \sigma$ of an object $B$ there is a morphism $f \colon A \to B$ with
\[
\scalebox{0.85}{\input{figures/conditioning1.tikz}}

\quad
\text{ and }
\qquad
\scalebox{0.85}{\input{figures/conditioning2.tikz}}

\]
\end{principle}

Intuitively, this principle simply states that there is some process $f$ which prepares either the state $\rho$ or $\sigma$ depending on whether the system is given in the state $\ket{0}$ or $\ket{1}$ (and may behave in any manner otherwise). The ability to consider such conditional procedures is a typical background assumption in the study of (causal) probabilistic theories, see e.g.~\cite[p.12]{chiribella2010purification}.

\begin{definition}[Operational Principles]
We say that a dagger theory $\catC$ satisfies the \emph{operational principles} when it is non-trivial, dagger compact, has zero morphisms, and satisfies principles 1-4.
\end{definition}

Along with these principles, we will later introduce some mild extra assumptions on the scalars in our theory, most notably `boundedness'. 

\begin{example} \label{ex:QuantCAndRearly}
$\Quant{\mathbb{C}}$ and $\Quant{\mathbb{R}}$ are dagger theories satisfying the operational principles, as we'll prove in Section~\ref{sec:PureProcessProperties}. The dagger theories $\Class$ and $\Rel$ may be seen to satisfy all of the operational principles aside from \strongpurification. 
\end{example}

\section{Consequences of the Principles} \label{sec:consequences}

Let us now establish some basic consequences of the operational principles. 

We begin with kernels. First, note that any category with dagger kernels in fact comes with some further useful maps. For any morphism $f \colon A \to B$, we define the \emph{image} of $f$ by
\[
\img(f) := \ker(\coker(f)) \colon \Img(f) \to B
\]
Then $f$ factors as $f = \img(f) \circ g$ where $\coker(g) = 0$, and a morphism $k$ is a kernel iff $k = \img(k)$ \cite{heunen2010quantum}. Dually, the \emph{coimage} of a morphism is given by $\coim(f) := \coker(\ker(f)) \colon A \to \CoIm(f)$. 

\begin{lemma} \label{lem:dag-ker-in-compact}
In any dagger compact category $\catC$ with dagger kernels:
\begin{enumerate} 
\item \label{enum:ker-extra-rule}
Kernels satisfy the rule 
\begin{equation} \label{eq:kernels}
\scalebox{0.85}{\input{figures/kernel-compact-intro1.tikz}}
 \ \implies \exists ! h  \text{ s.t. } \ 
\scalebox{0.85}{\input{figures/kernel-compact-intro2.tikz}}

\end{equation}
\item \label{enum:ker_under_tensor}
If $k$ and $l$ are kernels then so is $k \otimes l$;
\item \label{enum:zero-cancell}
If $\catC$ is state-inhabited then all morphisms $f, g$ satisfy
\begin{equation} \label{eq:zero-cancel}
f \otimes g = 0 \implies f = 0 \text{ or } g = 0
\end{equation}
\end{enumerate}
\end{lemma}
\begin{proof}
\ref{enum:ker-extra-rule} follows from bending wires and using the definition of a kernel.

\ref{enum:ker_under_tensor}. 
It suffices to show for any kernel $k \colon K \to A$ and object $B$ that $k \otimes \id{B} = \img(k \otimes \id{B})$ i.e.~is a kernel. Then the morphism $k \otimes l = (k \otimes \id{}) \circ (\id{K} \otimes l)$ is a composition of kernels, and hence a kernel~\cite{heunen2010quantum}, as required.

 Now suppose that $f \colon C \to A \otimes B$ has that $g \circ f = 0$ for any $g \colon A \otimes B \to D$ orthogonal to $k \otimes \id{{B}}$. Then composing $f$ with $(k^\bot \otimes \id{B})$ and using the previous part shows that $f$ factors over $k \otimes \id{B}$, as desired.

\ref{enum:zero-cancell}. Using compactness we have
\[
\scalebox{0.85}{\input{figures/kernel-bend1.tikz}}
 \ 0
\iff 
\scalebox{0.85}{\input{figures/kernel-bend2.tikz}}
 \ 0
\]
Suppose that $f \otimes g = 0$ and let the effect $e$ be given by bending the output of $g$ to an input as above. If $\Img(e) \simeq 0$ then since $e$ factors over $\Img(e)$ we have $e = 0$. Otherwise, let $\psi$ be an isometric state of $\Img(e)$. Then $\img(e) \circ \psi$ is an isometry $I \to I$ and hence is unitary since scalars are commutative. So $\coker(e) = 0$ and then by the above we have $f = 0$. 
\end{proof}

We next explore some consequences of the presence of purification and kernels. 

\begin{proposition} \label{prop:OpTheoryProperties}
Let $\catC$ be a non-trivial compact dagger theory satisfying principles 1 and~2. Then the following hold.
\begin{enumerate} 
\item \label{enum:zero-mon}
$\discard{} \circ f = 0 \implies f = 0$ for all morphisms $f$.
\item \label{enum:dag-ker-pure}
Every dagger kernel in $\catC$ is pure and causal and is a kernel in $\catC_{\pure}$.
\item \label{enum:pairOfPure}
Every non-trivial object has an orthonormal pair of pure states.
\end{enumerate}
\end{proposition}
\begin{proof}
\ref{enum:zero-mon}.
Let $g$ be a purification of $f$. Then $\discard{} \circ g = 0 = \discard{} \circ 0$. Then since $0$ is pure by definition we have $g = U \circ 0$ for some unitary $U$. Then $g = 0$ and so $f = 0$ also.  

\ref{enum:dag-ker-pure}.
We claim that any kernel $k$ satisfies~\eqref{eq:whole}. Indeed if $g$ is a dilation of $k$ then by the previous part we have implications:
\[
\scalebox{0.85}{\input{figures/kernelsarepurei.tikz}}

\]
for some unique $f$. Then since $k$ is an isometry, $f$ satisfies the left hand side of~\eqref{eq:noleaks}. If $K \simeq 0$ then $k = 0$ and so it is pure, otherwise $f = \id{A} \otimes \rho$ for some state $\rho$, as required. Hence all kernels are pure.

Now for any pure morphism $f$, $k = \ker(f)$ is also the kernel of $f$ in $\catC_{\pure}$, since if $g$ is pure with $f \circ g = 0$, then $g = k \circ h$ where $h = k^{\dagger} \circ g$ is pure since $g$ and $k$ are.

\ref{enum:pairOfPure}.
Let $A$ be any non-trivial object in $\catC$, and $\psi$ any causal pure state of $A$. By pure exclusion, $\Coker(\psi)$ is non-zero and so has a causal pure state $\eta$. Then $\phi = \coker(\psi)^{\dagger} \circ \eta$ is also a causal pure state of $A$ and $\psi$ and $\phi$ are orthonormal.
\end{proof}

Next let us turn to some consequences of pure exclusion, which in the presence of our other principles has a natural equivalent form. Say that a dagger theory has \emph{normalisation} when it has that every non-zero state $\rho$ is of the form $\rho = \sigma \circ r$ for some causal state $\sigma$ and scalar $r$. For example this holds in theories whose scalars are $\Rpos$ or the Booleans $\mathbb{B}$.

\begin{lemma}\label{lem:exclusion-equiv-new}
The following are equivalent for a compact dagger theory satisfying principles 1 and 2:
\begin{enumerate}
\item \label{enum:PEa}
Pure exclusion holds;
\item \label{enum:trivImage}
Every pure state has trivial image;
\label{enum:ImPsitrivial}
\item \label{enum:NormAndExtraa}
Normalisation holds, and every causal pure state is a kernel.
\end{enumerate}
Moreover, when these hold:
\begin{enumerate}[label=\roman*), ref=\roman*]
\item \label{enum:scalpure}
All scalars are pure;
\item  \label{enum:puretensor}
For all non-zero morphisms $f, g$, if $f \otimes g$ is pure then so are $f$ and $g$.
\end{enumerate}
\end{lemma}
\begin{proof}
\ref{enum:PEa} $\implies$ \ref{enum:trivImage}: 
 Let $\psi$ be any non-zero pure state. Now we can write $\psi = \img(\psi) \circ \phi$ for some state $\phi$ with $\coker(\phi) = 0$. Since $\img(\psi)$ is an isometry it is easy to see that $\phi$ is pure also. Hence by pure exclusion $\Img(\psi)$ is trivial.

\ref{enum:trivImage} $\implies$ \ref{enum:PEa}: 
Let $\psi$ be a non-zero pure state of some object $A$. By assumption $\psi$ has trivial image and so we may take $\Img(\psi) = I$. Now suppose that $e \circ \psi = 0$ for all effects $e$. Thanks to Proposition~\ref{prop:OpTheoryProperties}~\eqref{enum:zero-mon}, we have $\coker(\psi) = 0$. This makes $\img(\psi)$ an isomorphism $I \simeq A$, which is causal by Proposition~\ref{enum:dag-ker-pure}, making $A$ trivial. 

\ref{enum:trivImage} $\implies$ \ref{enum:NormAndExtraa}: 
By purification, for normalisation it suffices to be able to normalise any non-zero pure state. By assumption any such state $\psi$ has $\Img(\psi) = I$, so that
\[
\scalebox{0.85}{\input{figures/img-pic.tikz}}

\]
and since $\img(\psi)$ is a kernel it is causal. Moreover if $\psi$ is causal then by applying discarding we see that $r = \id{I}$, making $\psi=\img(\psi)$ a kernel.  

\ref{enum:NormAndExtraa} $\implies$ \ref{enum:scalpure}: Normalisation is precisely the statement that all scalars are pure.

\ref{enum:NormAndExtraa} $\implies$ \ref{enum:puretensor}: By compactness, it suffices to show that if $\sta \otimes \stb$ is pure for non-zero states $\sta, \stb$, of $A, B$ respectively, then so is $\stb$. 

We first consider when $\sta$ is causal. In this case suppose that some state $\rho$ of $B \otimes C$ dilates $\stb$. Then $\sta \otimes \rho$ dilates the pure state $\sta \otimes \stb$ and so is given by  $\sta \otimes \stb \otimes \nu$ for some causal state $\nu$ of $C$. But then 
\[
\scalebox{0.85}{\input{figures/closed-under-statesargii.tikz}}

\]
Hence $\stb$ is pure. 

In the general case, by normalisation we have $\sta = \tau \circ r$ for some scalar $r$ and causal state $\tau$. Then since $\tau$ is causal, by the case above the state defined by $\stc := \stb \circ r$ must be pure. Now any purification $\sigma \colon I \to B \otimes D$ of $\stb$ satisfies 
\[
\scalebox{0.85}{\input{figures/imgarg-1.tikz}}

\]
for any causal pure state $\mu$ of $D$. Since we've seen that every scalar is pure, by essential uniqueness we have
\[
\scalebox{0.85}{\input{figures/imgarg-new.tikz}}

\]
for some pure unitary $U$. Since any unitary is causal by the CP axiom, letting $\eta$ be the causal pure state $U \circ \mu$ we have implications
\[
\scalebox{0.85}{\input{figures/imgarg-2.tikz}}
 \ 0 \  \ 
\implies
\scalebox{0.85}{\input{figures/imgarg-2a.tikz}}
 \ 0 \  \ 
\implies
\scalebox{0.85}{\input{figures/imgarg-2b.tikz}}

\]
for some state $\stl$, since by assumption $\eta$ is a kernel and so $\eta = \img(\eta)$. Note that the first implication uses Lemma~\ref{lem:dag-ker-in-compact}~\eqref{enum:zero-cancell}. Then applying $\discard{D}$ gives $\stl = \stb$.  Hence $\stb \otimes \eta$ is pure. Since $\eta$ is causal, we have seen above that this makes $\stb$ pure as required.  

{\ref{enum:NormAndExtraa}} $\implies$ \ref{enum:trivImage}: Let $\psi$ be any non-zero pure state, and let $\psi = \phi \circ r$ for some scalar $r$ and causal state $\phi$. Then by Lemma~\ref{lem:dag-ker-in-compact}~\eqref{enum:zero-cancell} $\coker(r) = 0$, so that $\img(\psi) = \img(\phi)$. But since $\phi$ is pure by \eqref{enum:puretensor}, it is a kernel by assumption and so $\img(\phi)$ is trivial. 
\end{proof}

Another useful fact for us will be the following.

\begin{lemma} \label{lem:specialObjectC}
Any dagger theory satisfying the operational principles contains an object with a pair of causal pure states $\ket{0}, \ket{1}$ with $\ket{0} = \ket{1}^\bot$ as kernels. 
\end{lemma}
\begin{proof}
Let $A$ have a pair of orthonormal pure states $\psi_0, \psi_1$, as in Proposition~\ref{prop:OpTheoryProperties}~\eqref{enum:pairOfPure}. As noted, the dagger kernels on $A$ form an orthomodular lattice, so we may define 
\[
B = \Img(\psi_0) \vee \Img(\psi_1)
\] 
and $i := \img(\psi_0) \vee \img(\psi_1) \colon B \to A$. Then $\psi_0 = i \circ \ket{0}$ and $\psi_1 = i \circ \ket{1}$ for some unique pure isometric states $\ket{0}, \ket{1}$, which are kernels by the CP axiom and Lemma~\ref{lem:exclusion-equiv-new}. Furthermore $\ket{0}$, $\ket{1}$ are then orthogonal and by construction satisfy $\ket{0} \vee \ket{1} = \id{B}$. By general properties of orthomodular lattices this means that $\ket{0} = \ket{1}^{\bot}$. 
\end{proof}

\subsection{Coarse-graining}

In fact the principles we have considered provide another basic feature typical to operational physical theories. 

\begin{definition}
A dagger theory $\catC$ has \emph{coarse-graining} when it is dagger monoidally enriched in commutative monoids. That is, it has zero morphisms and each homset $\catC(A,B)$ comes with an associative commutative operation $+$ such that 
\begin{align*}
f \circ (g + h) &= f \circ g + f \circ h 
&
f + 0 &= f
\\ 
f \otimes (g + h) &= f \otimes g + f \otimes h 
&
(f + g)^{\dagger} &= f^{\dagger} + g^{\dagger} 
\end{align*}
for all morphisms $f, g, h$.
\end{definition}

\begin{examples} \label{example:CPM-mixing}
Our classical examples $\Class$ and $\Rel$ each come with coarse-graining given by addition of matrices and union of relations, respectively. If $\catA$ is a dagger compact category with finite dagger biproducts then $\CPM(\catA)$ has coarse-graining defined by
\[
\scalebox{0.85}{\input{figures/CPM-map-smaller.tikz}}

+
\scalebox{0.85}{\input{figures/CPM-map-smaller2.tikz}}

:=
\scalebox{0.85}{\input{figures/CPM-map-smaller3.tikz}}

\]
where $\langle f, g \rangle$ is the unique such morphism with 
$(\coproj_C^{\dagger} \otimes \id{B}) \circ \langle f, g \rangle = f$ and $(\coproj_D^{\dagger} \otimes \id{B}) \circ \langle f, g \rangle = g$.
In particular $\Quant{S}$ has coarse-graining, for any commutative involutive semi-ring $S$.
\end{examples}

\begin{proposition} \label{prop:coarse-graining}
In the presence of the other operational principles, $\catC$ satisfies conditioning iff it has coarse-graining. In this case its coarse-graining operation is unique. Moreover, in a theory with coarse-graining, causal complementation holds iff all dagger kernels $k \colon K \to A$ are causal and satisfy 
\begin{equation} \label{eq:causally-complemented}
\scalebox{0.85}{\input{figures/causal-comp.tikz}}

\end{equation}
\end{proposition}

\begin{proof}
Suppose that $\catC$ has coarse-graining. Then conditioning follows automatically, since $f \circ \ket{0} = \rho$ and $f \circ \ket{1} = \sigma$ is satisfied by setting 
\[
\scalebox{0.85}{\input{figures/condition-dagger.tikz}}

\]
Conversely, suppose that $\catC$ satisfies the operational principles. By Lemma~\ref{lem:specialObjectC} it contains an object $C$ with a pair of causal kernel states $\ket{0}, \ket{1}$ with $\ket{0} = \ket{1}^\bot$. Now given any $f, g \colon A \to B$, using conditioning and compactness let $h \colon A \to B \otimes C$ be any morphism with 
\begin{equation} \label{eq:cg-def1}
\scalebox{0.85}{\input{figures/cg-1.tikz}}

\end{equation}
We then define 
\begin{equation} \label{eq:cg-def2}
\scalebox{0.85}{\input{figures/cg-2i.tikz}}

\end{equation}
By causal completeness this is independent of our choice of $h$. Moreover it is straightforward to verify that it respects $\circ, \otimes$ and $\dagger$ and has unit $0$, and so indeed gives $\catC$ coarse-graining.

For the second statement, firstly note that~\eqref{eq:causally-complemented} is easily seen to ensure causal complementation. Conversely, for any kernel $k$ as above let $f = k \circ k^\dagger + k^{\bot} \circ k^{\bot \dagger} \colon A \to A$. Then we have 
\[
\scalebox{0.85}{\input{figures/cg-8.tikz}}

\]
and so by causal complementation, $f$ is causal. But since all dagger kernels are causal we have $\discard{} \circ f 
= 
\discard{} \circ k^\dagger + \discard{} \circ k^{\bot \dagger}$
and so~\eqref{eq:causally-complemented} holds. 

Finally let us show that $+$ as defined above is unique. Indeed if $\catC$ comes with any other such operation then by~\eqref{eq:causally-complemented} for any object $C$ as above we have 
\[
\scalebox{0.85}{\input{figures/cg-7.tikz}}

\]
It follows that any morphism $h$ satisfying~\eqref{eq:cg-def1} will then automatically have marginal $f + g$, and so $+$ coincides with our definition above. In particular $+$ is independent of our choice of $C$.
\end{proof}

In particular, the operation $+$ makes the scalars $R = \catC(I,I)$ of our theory into a semi-ring. Thanks to the above we may have equivalently originally stated the operational principles in terms of the presence of such a coarse-graining operation; for more on this see Appendix~\ref{append:equiv}. 

\section{Deriving Superpositions} \label{sec:derivingsups}

Let us now show that any theory $\catC$ satisfying our principles has the superposition-like features of Section~\ref{sec:building}.

Firstly, note that thanks to essential uniqueness every pair of causal pure states of the same object in our theory are connected by a unitary. In fact we are able to strengthen this property. 

\begin{lemma} \label{lem:strong_homog}
In any dagger theory satisfying the operational principles, for any pairs $\{\ket{0}, \ket{1}\}$ and $\{\ket{0'}, \ket{1'}\}$ of orthonormal pure states of an object $A$, there is a unitary $U$ on $A$ with $U \circ \ket{0} = \ket{0'}$ and $U \circ \ket{1} = \ket{1'}$. 
\end{lemma}

\begin{proof}
By essential uniqueness there is a unitary $U$ on $A$ with $U \circ \ket{0} = \ket{0'}$. Since every causal pure state is a dagger kernel we may define a new dagger kernel $k = \ket{0'}^\bot \colon K \to A$. 

Since unitaries preserve orthogonality, $U \circ \ket{1}$ is orthogonal to $\ket{0'}$, so that $U \circ \ket{1}  = k \circ \psi$ for the causal pure state $\psi = k^\dagger \circ U \circ \ket{1}$. Similarly we always have $\ket{1'} = k \circ \phi$ for some causal pure state $\phi$. By essential uniqueness there is then a unitary $V$ on $K$ with $V \circ \psi = \phi$, and in turn a unitary $W$ on $A$ with $W \circ k = k \circ V$. 

One may then verify that $W^{\dagger} \circ \ket{0'}$ is orthogonal to $k$ and so factors over $k^{\bot} = \ket{0'}^{\bot \bot} = \ket{0'}$. Hence we have $W^{\dagger} \circ \ket{0'} = \ket{0'} \circ z$ for some scalar $z$. Then since $\ket{0'}$ is an isometry so is the scalar $z$, and so, since all scalars are pure, by the CP axiom we have $z = \id{I}$. Finally since $W$ preserves $\ket{0'}$ we have that $W \circ U$ is the desired unitary on $A$.
\end{proof}

In just the same way one may in fact show that there is a unitary relating any two collections of orthonormal states of the same size $\{\ket{i}\}^n_{i=1}$ and $\{\ket{i'}\}^n_{i=1}$; this property is called `strong symmetry' in~\cite{barnum2014higher}.
The result also allows us to extend conditioning to pure morphisms as follows.

\begin{lemma} \label{lem:existence}
In any dagger theory satisfying the operational principles, for any orthonormal pure states $\ket{0}, \ket{1}$ of an object $A$ and pair of pure states $\psi, \phi$ of an object $B$ there is a pure $f \colon A \to B$ with $f \circ \ket{0} = \psi$ and $f \circ \ket{1} = \phi$. 
\end{lemma}
\begin{proof}
If $\psi = 0$ then we may take $f = \phi \circ \ket{1}^\dagger$, and similarly if $\phi = 0$ the result is trivial. Otherwise assume that $\psi$ and $\phi$ are non-zero. Using conditioning, let $h$ be any morphism satisfying 
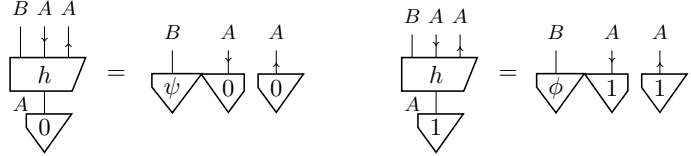
\[
\scalebox{0.85}{\input{figures/purif-in-proof-program-depure.tikz}}

\]
and let $g$ be any purification of $h$ via some object $C$. Then since all morphisms involved are pure it follows that 
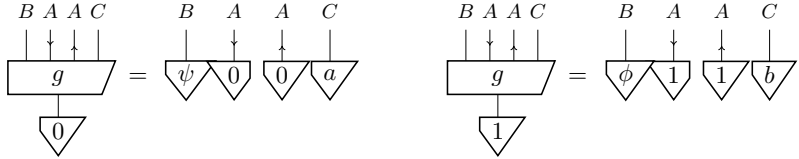
\[
\scalebox{0.85}{\input{figures/purif-in-proof-program-2.tikz}}

\]
for some causal states $a, b$. By Lemma~\ref{lem:exclusion-equiv-new}~\eqref{enum:puretensor} these must be pure since each of the left-hand sides above are. Hence by Lemma~\ref{lem:strong_homog} there is a unitary $U$ with 
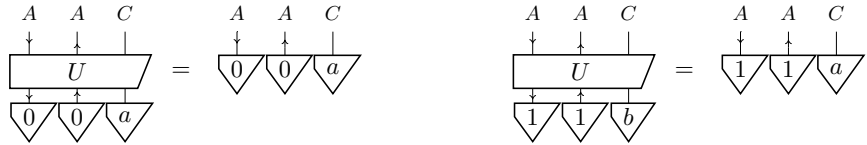
\[
\scalebox{0.85}{\input{figures/U-function.tikz}}

\]
Finally the pure morphism defined by
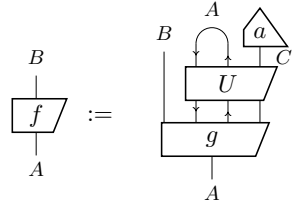
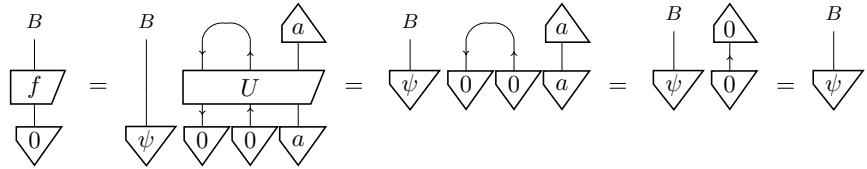
\[
\scalebox{0.85}{\input{figures/proof-fancy-construction.tikz}}

\]
satisfies
\[
\scalebox{0.85}{\input{figures/proof-fancy-long.tikz}}

\] and $f \circ \ket{1} = \phi$ similarly. 
\end{proof}

We are now able to show that $\catC_{\pure}$ has a qubit-like object.

\begin{proposition} \label{prop:PhBiprodII}
Let $\catC$ be a dagger theory satisfying the operational principles. Then
$\catC_{\pure}$ has a phased dagger biproduct $B = I \pbiprod I$ for which all phases are unitary. 
\end{proposition}

\begin{proof}
Let $B$ be any object with a pair of pure causal states $\ket{0}, \ket{1}$ with $\ket{0} = \ket{1}^\bot$ as dagger kernels, as in Lemma~\ref{lem:specialObjectC}. Then $\ket{0}, \ket{1} \colon I \to B$ satisfy the existence property of a phased biproduct in $\catC_\pure$ by Lemma~\ref{lem:existence}. 

We now establish the uniqueness property. Let $\tinymultflip[whitedot] \colon B \to B \otimes B$ be a pure morphism with $\tinymultflip[whitedot] \circ \ket{i} = \ket{i} \otimes \ket{i}$ for $i = 0,1$. Then 
\[
\scalebox{0.85}{\input{figures/copier-zero-1.tikz}}
 \  0 
\]
Hence the morphism above the state $\ket{1}$ on the left-hand side above factors over $\coker(\ket{1}) = \ket{0}^\dagger$, and so is invariant under composition with  $\ket{0} \circ \ket{0}^\dagger$ since $\ket{0}$ is an isometry. In other words we have 
\begin{equation} \label{eq:zero-eq}
\scalebox{0.85}{\input{figures/copier-top-0.tikz}}

\end{equation}
The corresponding equation for $\ket{1}$ holds similarly.

Now let $f, g \colon B \to A$ be pure with $f \circ \ket{i} = g \circ \ket{i}$ for $i= 0, 1$. If $f \circ \ket{0} = 0$ then  since $\ket{1} = \ket{0}^\bot$ we get $f = f \circ \ket{1} \circ \ket{1}^\dagger = g$, and similarly $f = g$ if $f \circ \ket{1} = 0$. So now suppose that $f \circ \ket{i} \neq 0$ for $i = 0,1$. By \eqref{eq:zero-eq} we have 
\[
\scalebox{0.85}{\input{figures/copier-EUPurif-1a.tikz}}

\]
and so bending wires and using causal complementation we get
\[
\scalebox{0.85}{\input{figures/copier-EUPurif-1.tikz}}

\quad
\text{ and so }
\quad
\scalebox{0.85}{\input{figures/copier-EUPurif2.tikz}}

\]
for some unitary $U$ by essential uniqueness. But then 
\[
0 \ \ 
\scalebox{0.85}{\input{figures/copier-fapply-merge.tikz}}

\]
Hence by zero cancellativity we have $\ket{1}^\dagger \circ U \circ \ket{0} = 0$ and so $U \circ \ket{0} = \ket{0} \circ z$ for some scalar $z$. But then $z$ is an isometry and so by the CP axiom $z = \id{I}$. Hence $U$ preserves the states $\ket{0}$, and $\ket{1}$, i.e.~is a phase. Since $B$ has the existence property of a phased biproduct, there exists a pure effect $\tinycounit$ on $B$ with $\tinycounit \circ \ket{0} = \id{I} = \tinycounit \circ \ket{1}$. Then for any such effect $\tinycounit$ we have 
\begin{equation} \label{eq:witjcounit}
\scalebox{0.85}{\input{figures/copier-last.tikz}}
 
\end{equation}
where each of the endomorphisms of $B$ below $f, g$ above are also phases.

Finally note that any phase $W$ is unitary, since we have that $\ket{i}^\dagger \circ W^\dagger = \ket{i}^\dagger$ for $i =0,1$ and so $W^\dagger$ is causal by causal complementation, and hence unitary by essential uniqueness. In particular this makes phases invertible, so that by~\eqref{eq:witjcounit} $f$ and $g$ are equal up to phase, and closed under the dagger, so that $B$ is a phased dagger biproduct.      
\end{proof}

To construct more general phased dagger biproducts in $\catC_{\pure}$ we use the following abstract results, each proven in Appendix~\ref{appendix:Proofs}.

\begin{lemma} \label{lem:IplusIgivesaplusa}
Let $\catB$ be a dagger compact category with zero morphisms and a phased dagger biproduct $I \pbiprod I$. Then $\catB$ has phased dagger biproducts of the form $A \pbiprod A$, for all objects $A$.
\end{lemma}

\begin{lemma} \label{lem:kernelsGiveBiproducts}
Let $\catB$ be a category with dagger kernels and phased dagger biproducts $A \pbiprod A$ for all objects $A$. Then $\catB$ has finite phased dagger biproducts iff for every pair of objects $A, B$ there is an object $C$ and orthogonal kernels $k \colon A \to C$ and $l \colon B \to C$.  
\end{lemma}

\begin{corollary} \label{cor:daggerphasedbiproducts}
Let $\catC$ be a dagger theory satisfying the operational principles. Then $\catC_{\pure}$ has the superposition properties.
\end{corollary}

\begin{proof}
Since all kernels in $\catC$ are also kernels in $\catC_\pure$, by Proposition~\ref{prop:PhBiprodII} and Lemmas~\ref{lem:IplusIgivesaplusa} and~\ref{lem:kernelsGiveBiproducts} to show that $\catC_\pure$ has phased dagger biproducts it suffices to show for all objects $A$, $B$ that there exists some object $C$ and orthogonal kernels $k \colon A \to C$ and $l \colon B \to C$.

Now if either $A$ or $B$ is a zero object the result is trivial. Otherwise let $\psi$ and $\phi$ be causal pure states of $A, B$ respectively, and let $C$ be an object with two orthogonal causal pure states $\ket{0}$, $\ket{1}$, such as $I \pbiprod I$. Then these states are all kernels and so by Proposition~\ref{lem:dag-ker-in-compact}~\eqref{enum:ker_under_tensor} so are the morphisms
\[
\scalebox{0.85}{\input{figures/ker1.tikz}}

\]
which are indeed orthogonal. Hence $\catC_\pure$ has finite phased dagger biproducts. 

We now verify the rule~\eqref{eq:positive-cancellation}. First consider a pure positive endomorphism $f^{\dagger} \circ f$ of $A \pbiprod B$ which is diagonal so that $f_A := f \circ \pcoproj_A$ and $f_B := f \circ \pcoproj_B$ are orthogonal. Letting $c_A = \img(f_A)^{\dagger}$ and $c_B = {c_A}^{\bot}$ we have
\begin{align*}
f_A^\dagger \circ f_B = 0 
\implies 
c_A \circ f_B = 0
&\implies 
c_A \circ f = c_A \circ f \circ \coproj_A \circ \coproj_A^\dagger 
\\ 
c_B \circ f \circ \coproj_A
=
\coker(f_A) \circ f_A
=
0
&\implies
c_B \circ f = c_B \circ f \circ \coproj_B \circ \coproj_B^\dagger 
\end{align*}
using for the latter implications that $\pcoproj_A$ and $\pcoproj_B$ are themselves dagger kernels with $\pcoproj_A=\pcoproj_B^\bot$, as shown in Lemma~\ref{lem:dag-ker-coproj} in the Appendix. Hence we have 
\begin{align} \label{eq:poscancel2a}
\discard{} \circ c_A \circ f 
&=
\discard{} \circ f_A \circ \coproj_A^{\dagger}
\\   \label{eq:poscancel2b}
\discard{} \circ c_B \circ f 
&=
\discard{} \circ f_B \circ \coproj_B^{\dagger}
\end{align}
Now if any other pure diagonal endomorphism $g$ has $f^{\dagger} \circ f = g^{\dagger} \circ g \circ U$ for some phase $U$, defining $g_A = g \circ \coproj_A$ and $g_B = g \circ \coproj_B$ we have that $f_A^{\dagger} \circ f_A = g_A^{\dagger} \circ g_A$, and the similar equation holds for $B$. Then by the CP axiom 
\begin{align*}
\discard{} \circ f_A = \discard{} \circ g_A 
& & 
\discard{} \circ f_B = \discard{} \circ g_B
\end{align*}
Since $c_B = {c_A}^{\bot}$, by causal complementation and~\eqref{eq:poscancel2a}, \eqref{eq:poscancel2b} we have $\discard{} \circ f = \discard{} \circ g$. Finally $f^{\dagger} \circ f = g^{\dagger} \circ g$ by the CP axiom again.
\end{proof}


\section{The Extended Pure Morphisms} \label{sec:PureProcessProperties}

We will be able to learn more about any theory satisfying our principles by studying its subcategory of pure morphisms further. In fact it will be more fruitful to consider the extension of this subcategory via the $\GP$ construction, just as $\FHilb$ is typically considered in place of $\FHilb_\sim$. More precisely, this section concerns categories of the following form.

\begin{definition} \label{def:quantumcat}
A \emph{quantum category} is a dagger compact category $\catA$ with finite dagger biproducts and dagger kernels which is state-inhabited and satisfies:
\begin{itemize}
\item 
\emph{dagger normalisation}: every non-zero state $\psi$ has $\psi = \sigma \circ r$ for some isometric state $\sigma$ and scalar $r$;
\item 
\deff{homogeneity}: for all $f, g \colon A \to B$ we have
\[
\scalebox{0.85}{\input{figures/homog.tikz}}

\]
for some unitary $U$ on $B$.
\end{itemize}
\end{definition}

\begin{proposition} \label{prop:QuantWOGToQuantBiprod}
Let $\catC$ be a dagger theory satisfying the operational principles. Then $\catA=\plusI{\catC_\pure}$ is a quantum category. 
\end{proposition}

\begin{proof}
We've seen that $\catC_\pure$ has the superposition properties and kernels. 
Now by Theorem~\ref{thm:onPHConstr} $\catA$ is dagger compact with finite dagger biproducts, and we may identify $\catC_\pure$ with the category $\catA_\sim$ of equivalence classes $[f]$ of its morphisms under $f \sim g$ whenever $f = u \cdot g$ for some unitary `global phase' scalar $u \in \mathbb{P}$.

By the last line of Theorem~\ref{thm:onPHConstr} we have $[f^\dagger \circ f] = \id{} \implies f^\dagger \circ f = \id{}$ for all morphisms $f \in \catA$. In particular a morphism $f$ in $\catA$ is an isometry or unitary iff $[f]$ is in $\catC_\pure$. This lets one deduce dagger normalisation, homogeneity and state-habitation noting that they hold in $\catC_\pure$ by strong purification in $\catC$.

Finally, noting that $[f] = 0 \implies f = 0$ it's easy to see that if $[f] = \ker([g])$ in $\catC_\pure$ then $f = \ker(g)$ in $\catA$, giving the latter category kernels also.
\end{proof}

\begin{examples}
$\FHilb$ is a quantum category, as follows from the previous result; alternatively all there is left to note is homogeneity, which follows from the polar decomposition of a complex matrix.

In contrast homogeneity fails in $\Hilb$; for example on $l^2(\mathbb{N})$ the shift operator $(a_0, a_1, \dots) \mapsto (0, a_0, \dots)$ is an isometry but not unitary.
\end{examples}

Quantum categories have a surprisingly rich structure, generalising that of $\FHilb$, which we now explore. Recall that since they have biproducts they come with an addition operation $f + g$ on morphisms, generalising superpositions (and not to be confused with coarse-graining in a dagger theory). In fact they also come with a notion of subtraction.

\begin{proposition} \label{prop:quant-cat-properties}
In any quantum category $\catA$, the following hold.
\begin{enumerate} 
\item \label{enum:purepropnegatives}
Every morphism $f$ has an additive inverse $-f$;
\item \label{enum:daggerequalisers}
Every pair of morphisms $f, g$ have a \emph{dagger equaliser}~\cite{vicary2011categorical};
\item \label{enum:dagmonodagkernel}
Every isometry is a kernel; 
\item \label{enum:decomp-property}
For every kernel $k \colon K \to A$ the morphism $[k, k^{\bot}] \colon K \biprod K^{\bot} \to A$ is unitary;
\item \label{enum:wellpted}
\emph{Well-pointedness}: ($f \circ \psi = g \circ \psi$ $\forall$ states $\psi$) $\implies$ $f =g$.
\end{enumerate}
\end{proposition}

\begin{proof}
\ref{enum:purepropnegatives}.
It suffices to find a scalar $t$ with $t + \id{I} = 0$, since then for all $f$ we have $f + (t \cdot f) = (\id{I} + t) \cdot f = 0$. 
Employing standard categorical notation, we write $\langle a_1, a_2 \rangle$ for the unique state of $I \biprod I$ with $\pproj_i \circ \langle a_1, a_2 \rangle = a_i $, where $\pproj_i = \coproj_i^\dagger$, for $i=1,2$.
 Now let
\[
\begin{tikzcd}[column sep = huge]
I \rar{\Delta = \langle \id{I}, \id{I} \rangle} & I \biprod I
\end{tikzcd}
\]
  have $\Delta = \psi \circ s$ for some scalar $s$ and isometric state $\psi$. By homogeneity there is a unitary $U$ with $U \circ \pcoproj_1 = \psi$. Then let $\langle a, b \rangle = U \circ \pcoproj_2$. Since $U \circ \pcoproj_2$ is an isometry we have $a^\dagger \circ a + b^\dagger \circ b = \id{I}$ and also 
\begin{align*}
a + b 
&=
\Delta^{\dagger} \circ U \circ \pcoproj_2
\\ 
&= 
(s^{\dagger} \circ \pproj_1 \circ U^{\dagger}) \circ (U \circ \pcoproj_2)
\\ 
&= 
s^{\dagger} \circ (
\pproj_1^\dagger \circ \pcoproj_2 
)
= 
0
\end{align*}
Then $t = a^\dagger \circ b + b^\dagger \circ a$ is the required scalar since
\begin{align*}
\id{I} + t
&=
a^\dagger \circ a + b^\dagger \circ b + a^\dagger \circ b + b^\dagger \circ a 
\\
&=
(a + b)^{\dagger} \circ (a + b)
=
0
\end{align*}

\ref{enum:daggerequalisers}. This follows immediately with $f, g$ having dagger equaliser $\ker(f-g)$.

\ref{enum:dagmonodagkernel}. Thanks to~\eqref{enum:purepropnegatives} a morphism $m$ is monic iff $\ker(m) = 0$. But then any isometry $i$ has $i = \img(i) \circ e$ with $\coker(e) = 0$, and so dually $e$ is an epimorphism. But since $i$ and $\img(i)$ are isometries, so is $e$, and so it is unitary and $i$ is a kernel. 

\ref{enum:decomp-property}. 
$i = [k,k^{\bot}]$ is an isometry since $k$ and $k^{\bot}$ are orthogonal isometries. But if $f \circ i = 0$ then $f \circ k = 0$ and $f \circ k^{\bot} = 0$, so that $\img(f) = 0$ giving $f = 0$. Hence as in the previous part $i$ is epic, and so unitary.

\ref{enum:wellpted}.
Suppose that $f \circ \psi = g \circ \psi$ for all states $\psi$. Then $h = (f-g)$ has $h \circ \psi = 0$  and so $\coim(h) \circ \psi = 0$ for all states $\psi$. But if $\CoIm(h)$ is non-zero then it possesses a non-zero state $\phi$, and then $\coim(h) \circ \coim(h)^{\dagger} \circ \phi = \phi \neq 0$, a contradiction. Hence $\coim(h) = 0$ so that $h = 0$ and $f = g$.
\end{proof}

This provides the scalars in a quantum category $\catA$ with extra properties. 

\begin{definition} 
A \deff{phased ring} is a commutative involutive ring $(S, \dagger)$ which is an integral domain (meaning $a \cdot b = 0 \implies a = 0 \text{ or } b = 0$) 
such that $\forall a, b$
\[
  a^{\dagger} \cdot a + b^{\dagger} \cdot b = c^{\dagger} \cdot c
\] 
for some $c \in S$, with any such $c$ having $a = c \cdot d$ and $b = c \cdot e$ for some $d, e \in S$. By a \deff{phased field} we mean a phased ring which is also a field.
\end{definition}

In particular the positive elements $S^{\pos}$ of a phased ring are closed under addition, forming a sub-semi-ring of $S$. Examples of phased rings include $\mathbb{C}$, as well as $\mathbb{R}$ under the trivial involution.  

\begin{proposition} \label{prop:scalars_are_phased_ring}
Let $\catA$ be a quantum category. Then $\catA(I,I)$ is a phased ring.
\end{proposition}

\begin{proof}
$S = \catA(I,I)$ forms a commutative semi-ring since $\catA$ has dagger biproducts, and $S$ is a ring by Proposition~\ref{prop:quant-cat-properties}~\eqref{enum:purepropnegatives} and an integral domain by Lemma~\ref{lem:dag-ker-in-compact}~\eqref{enum:zero-cancell}. 

Now given $a, b \in S$ let $\psi = \langle a, b \rangle \colon I \to I \biprod I$. Using normalisation let $\psi = \phi \circ c$ where $\phi$ is an isometry. Then 
 \[
 c^\dagger \cdot c = \psi^{\dagger} \circ \psi = a^\dagger \cdot a + b^\dagger \cdot b
 \]
 Furthermore $d = \pproj_1 \circ \phi \in S$ has $a = \pproj_1 \circ \psi = c \cdot d$, and similarly $c$ divides $b$. Moreover any other $e \in S$ with $e^{\dagger} \cdot e = c^{\dagger} \cdot c$ has $e = c \cdot u$ for a unitary $u$ by homogeneity, and so also divides $a$ and $b$.
\end{proof}

\noindent
This provides our main examples of quantum categories, as we prove in Appendix~\ref{appendix:Proofs}. 

\begin{example} \label{example:MatofPhasedField}
Let $S$ be a phased field. Then $\Mat_S$ is a quantum category. In particular so are $\Mat_{\mathbb{C}}$ and $\Mat_{\mathbb{R}}$. 
\end{example}

It is useful to know a converse result, telling us when a quantum category $\catA$ arises as such a matrix category. In fact, for this we only require the following mild condition. Let us call a semi-ring $R$ \deff{bounded} when no element $r$ has that for all $n \in \mathbb{N}$ there is some $r_n \in R$ with $r = n + r_n$. For example the positive reals $\Rplus$ and rationals $\mathbb{Q^+}$ are certainly bounded. Boundedness is similar to the \emph{Archimedean} property for totally ordered groups~\cite{EncyMath}. 

\begin{lemma} \label{lem:Boundedequivalence}
Let $\catA$ be a quantum category and $S$ its ring of scalars. If $S^{\pos}$ is bounded then $\catA \simeq \Mat_{S}$.
\end{lemma}
\begin{proof}
Consider the full embedding $\MatS \hookrightarrow \catA$ given by $n \mapsto n \cdot I$. We now show that any object $A$ has a unitary $A \simeq n \cdot I$ for some $n \in \mathbb{N}$. If $A$ is a zero object we are done, otherwise there is an isometry $\psi \colon I \to A$, which is a kernel by Proposition~\ref{prop:quant-cat-properties}. Then by the same result, letting $B = \coker(\psi_1)$ the morphism $[\psi^{\bot}, \psi] \colon B \biprod I \simeq A$ is unitary. Setting $B_1 = B$ and proceeding similarly we obtain a sequence of objects $B_1, B_2, \dots$ with $A \simeq B_n \biprod n \cdot I$ for each $n$. Then if $B_n \simeq 0$ for some $n$ we are done. Otherwise for all $n \in \mathbb{N}$ we have
\[
\scalebox{0.85}{\input{figures/dimAbig.tikz}}

\!
=
\ 
\scalebox{0.85}{\input{figures/dimBnbig.tikz}}

= \  
\scalebox{0.85}{\input{figures/dimBbig.tikz}}

+\  
\sum^n_{i=1}
\scalebox{0.85}{\input{figures/dimSum.tikz}}

= \  
\scalebox{0.85}{\input{figures/dimBbig.tikz}}

+
\
n
\]
contradicting boundedness.
\end{proof}

Let us now see how these categories correspond precisely to theories satisfying our principles. We call a quantum category \deff{non-trivial} when it has $\id{I} \neq 0$.

\begin{theorem} \label{thm:mainToQuantumCat}
There is a one-to-one correspondence between non-trivial:
\begin{itemize}
\item 
quantum categories $\catA$;
\item 
dagger theories $\catC$ satisfying the operational principles; 
\end{itemize}
up to equivalence, via $\catA = \plusI{\catC_\pure}$ and $\catC = \CPM(\catA)$. 
\end{theorem}

\begin{proof}
For any such $\catC$, $\catA=\plusI{\catC_\pure}$ is a quantum category by Proposition~\ref{prop:QuantWOGToQuantBiprod}, and must be non-trivial since otherwise $\id{I}=0$ in $\catC$, making every object there trivial. By Lemma~\ref{lem:CPMsCoincide} the functor $[-]$ extends to an equivalence of dagger theories $\CPM(\catA) \simeq \catC$. 

The converse direction requires that for any non-trivial quantum category $\catA$, $\CPM(\catA)$ satisfies the operational principles and is equivalent to $\plusI{\CPM(\catA)_\pure}$. This fact is not needed for our main reconstruction, but is proven in Proposition \ref{prop:quant-to-theory} in Appendix \ref{appendix:Proofs}.
\end{proof}

This is a strong result, since for general $\catC$ with an environment structure there may be many $\catA$ with $\catC \simeq \CPM(\catA)$. 

\begin{example}
If $S$ is a phased ring which is also a field, $\Quant{S}$ satisfies our principles. In particular so do $\Quant{\mathbb{C}}$ and $\Quant{\mathbb{R}}$.
\end{example}

\begin{remark}
In this section we focused on the category $\plusI{\catC_\pure}$, but $\catC_\pure$ may also be axiomatised similarly; see Appendix~\ref{sec:pre-quantum}.
\end{remark}

\section{Reconstructions} \label{sec:reconstructions-listed}

We now return to the setting of dagger theories $\catC$ satisfying our principles, and have reached our main result.

\begin{theorem} \label{thm:main-reconstruction}
Let $\catC$ be a dagger theory satisfying the operational principles and $R = \catC(I,I)$. There is an embedding of dagger theories
\begin{equation} \label{eq:main-embed-in-recons} 
\Quant{S} \hookrightarrow \catC
\end{equation}
which preserves coarse-graining, for some phased ring $S$ with $R \simeq S^{\pos}$ as semi-rings. Moreover if $R$ is bounded this is an equivalence of theories $\catC \simeq \Quant{S}$.
\end{theorem}

\begin{proof}
By Theorem~\ref{thm:mainToQuantumCat} there is an equivalence of dagger theories $\catC \simeq \CPM(\catA)$ where $\catA$ is a quantum category. Moreover, note that any such equivalence must preserve the coarse-graining operation, since it is unique by Proposition~\ref{prop:coarse-graining}. Hence it suffices to assume that $\catC$ is of this form, in which case .$S = \catA(I,I)$ is indeed a phased ring. By dagger normalisation in $\catA$ the scalars in $\CPM(\catA)$ are isomorphic to $S^\pos$, since for every state $\psi$ in $\catA$ we have 
\[
\scalebox{0.85}{\input{figures/Sposarg.tikz}}

\]
for some isometric state $\phi$ and scalar $s$ in $\catA$.

The embedding $\Mat_{S} \hookrightarrow \catA$ is an equivalence when $R$ is bounded thanks to Lemma~\ref{lem:Boundedequivalence}, and it induces the respective embedding or equivalence~\eqref{eq:main-embed-in-recons}. Since the former preserves biproducts, the latter preserves coarse-graining. 
\end{proof}

As a result, when $\catC$ satisfies the operational principles, its scalars $R = \catC(I,I)$ form the positive elements $S^{\pos}$ of a phased ring $S$, generalising the relationship between $\mathbb{R^+}$ and $\mathbb{C}$. This provides $R$ with nice properties; it has characteristic $0$ and that $a$ is divisible by $a + b$ for all $a, b$, hence coming with an embedding $\mathbb{Q}^+ \hookrightarrow R$. 

We may also freely extend $R$ to a ring $\dring{R}$, its \emph{difference ring}. Formally $\dring{R}$ consists of pairs $(a,b)$ of elements of $R$ after identifying $(a,b)$ with $(c,d)$ whenever $a + d = b + c$. Addition and multiplication are defined in the obvious way when interpreting $(a,b)$ as `$a-b$'. For example $\dring{\mathbb{R}^+} = \mathbb{R}$. 

Under a final extra condition we can show that $S$ resembles either the real or complex numbers. Say that a semi-ring $R$ has \emph{square roots} when every $a \in R$ has $a = b^2$ for some $b \in R$. For any ring $S$ we write $S[i]$ for the involutive ring with elements of the form $a + b \cdot i$ for $a, b \in S$, where $1 = -i^2 = i \cdot i^{\dagger}$, and we define $a^{\dagger} = a$ for all $a \in S$.

\begin{lemma} \label{lem:square roots}
Let $S$ be a phased ring for which $R = S^{\pos}$ has square roots. 
\begin{enumerate}
\item \label{enum:posAndUn}
Every non-zero $s \in S$ has $s = r \cdot u$ for a unique $r \in R$ and unitary $u \in S$.
\item \label{enum:TotOrdered}
$R$ is totally ordered under $a \leq b$ whenever $a + c = b$ for some $c \in R$.
\item \label{enum:equalsSa}
$\dring{R} \simeq S^{\selfadj} := \{s \in S \mid s^{\dagger} = s \}$\label{not:sadjoint}.
\item \label{enum:TwoOptions}
Either $S = S^{\selfadj}$ with trivial involution, or $S$ has square roots and $S = S^{\selfadj}[i]$.
\end{enumerate}
\end{lemma}

\begin{proof}
\ref{enum:posAndUn}. 
For uniqueness, suppose that $p \cdot u = l \cdot v$ for $p, l \in R$ and $u, v$ unitary. Then $p = l \cdot w$ where $w = v \cdot u^{-1}$ is unitary. So 
\[l \cdot w^{\dagger} = p^{\dagger} = p = l \cdot w\] Since $S$ is an integral domain, multiplication is cancellative so $w = w^{\dagger}$ and $w^2 = 1$. If $w = 1$ we are done, otherwise $w = -1$ and so $p + l = 0$. But then $p = l = 0$ by the definition of a phased ring. 
For existence, given $s \in S$ let $r = s^\dagger \cdot s \in S^\pos$. Since $S^\pos$ has square roots, $r = t^2$ for some $t \in S^\pos$. Then $s^\dagger \cdot s = t^\dagger \cdot t$, so $s = t \cdot u$ for some $u$, which is easily seen to be unitary.

\ref{enum:TotOrdered}. Let $s \in S^{\selfadj}$ be non-zero with $s = t \cdot u$ as above. Then $t \cdot u = t \cdot u^{\dagger}$ and so $u =u^{\dagger}$ giving $u = \pm 1$. Hence either $s \in R$ or $-s \in R$. Then $\forall a, b \in R$ either 
\[a - b \in R\quad \text{ or } \quad b - a \in R\] making $R$ totally ordered in the above manner. 

\ref{enum:equalsSa}. We may identify $\dring{R}$ with the set of elements $a-b \in S$ for $a, b \in R$. Then $\dring{R} \subseteq S^{\selfadj}$ always, but by the previous part $S^{\selfadj} \subseteq \dring{R}$. 

\ref{enum:TwoOptions}. Suppose that $S \neq S^{\selfadj}$. We will show that $S$ has square roots using techniques adapted from Vicary~\cite[Thm.~4.2]{vicary2011categorical}. Since $R$ has square roots, thanks to the first part it suffices to find a square root of any given unitary $u \in S$. For this, first suppose that for all $s \in S$ the element
\[
x_s := s + u \cdot s^\dagger
\]
is zero. Then putting $s = 1$ shows that $u = -1$. But then since $x_s = 0$ for all $s$ we have $S = S^{\selfadj}$, a contradiction. Hence there is $s \in S$ such that $x_s$ is non-zero. Letting $x_s = r \cdot v$ for a unitary $v$ and $r \in R$ we have
\[
r \cdot v^\dagger = x_s^\dagger = u \cdot x_s = r \cdot v \cdot u^\dagger
\]
so that $v^\dagger = v \cdot u^\dagger$ and $v^2 = u$ as desired. 

Now in particular, $-1$ has a unitary square root $i$. Finally note that $2$ is divisible in $R$ thanks to the embedding $\mathbb{Q}_{\geq 0} \hookrightarrow R$. Then for any $s \in S$ defining elements of $S^\selfadj$ by
\[
\mathsf{R}(s) := \frac{1}{2} \cdot (s + s^{\dagger}) \qquad \mathsf{I}(s) := \frac{i}{2}(s^{\dagger} - s)\]
we have $s = \mathsf{R}(s) + i \cdot \mathsf{I}(s)$ so that $S = S^{\selfadj}[i]$. 
\end{proof}

\noindent
This gives a more precise form of our reconstruction theorem. 

\begin{corollary} \label{cor:precise-reconstruction}
Let $\catC$ be a dagger theory satisfying the operational principles, and suppose that $R=\catC(I,I)$ has square roots and is bounded. Then $\catC$ is equivalent to 
\[
\Quant{\dring{R}} \quad \text{ or } \quad \Quant{\dring{R}[i]}
\]
\end{corollary}

\section{Probabilistic Theories} \label{sec:probTheories}
Let us now consider the typical physical setting in which scalars correspond to (unnormalised) probabilities. 

\begin{definition} 
We call a dagger theory with coarse-graining $\catC$ \deff{probabilistic} when it comes with an isomorphism of semi-rings $\catC(I,I) \simeq \mathbb{R}^+$.
\end{definition}

Note that this is a weaker definition than typical in the literature~\cite{Barrett2007InfoGPTs,PhysRevA.84.012311InfoDerivQT} since we have not made any assumptions relating to tomography, finite-dimensionality or topological closure. One may in fact identify such theories intrinsically, as in the following observation for which we thank John van de Wetering. 

\begin{lemma} \label{lem:prob}
Let $\catC$ be a dagger theory satisfying the operational principles. Then $\catC$ is probabilistic iff $R=\catC(I,I)$ has square roots, no infinitesimals, and that every bounded increasing sequence has a supremum.
\end{lemma}
\begin{proof}
Appendix~\ref{appendix:Proofs}. 
 \end{proof}

Now immediately from the definition of a probabilistic theory, and Corollary~\ref{cor:precise-reconstruction}, we have the following.

\begin{corollary}
Any dagger theory which satisfies the operational principles and is probabilistic is equivalent to $\Quant{\mathbb{R}}$ or $\Quant{\mathbb{C}}$.
\end{corollary}

Finally one may distinguish complex from real quantum theory by adding an extra principle such as \emph{local tomography}~\cite{Hardy2012Holism},  or that every phase of a phased biproduct in our pure subcategory has a square root. In future it would be desirable to find a more generic categorical property separating these theories.

\section{Discussion} \label{sec:Discussion}

We close with some points of discussion.

\para{Equivalent axioms}
We saw that our principles ensured the presence of a coarse-graining operation $+$ on morphisms. Surprisingly, this operation is in fact cancellative, ruling out `possibilistic' theories like $\Rel$, as we show in Appendix~\ref{append:equiv}. Taking this operation instead as a primitive yields an alternative formulation of our principles, as follows.

Firstly, in a theory with coarse-graining, let us call a morphism $f$ \deff{sub-causal} when $\discard{} \circ f + e = \discard{}$ for some effect $e$. It is natural to consider such morphisms, since in fact only the sub-causal processes of a theory have a direct operational interpretation as possible outcomes of some experimental test. For example, in $\Quant{}$ a morphism is sub-causal when it is trace non-increasing as a completely positive map. 

\begin{theorem} \label{thm:alternate-axioms}
A dagger theory $\catC$ satisfies the operational principles iff it is non-trivial, dagger compact, has \strongpurification{} and dagger kernels, and comes with a coarse-graining operation for which: 
\begin{enumerate}
\item 
$f^{\dagger}$ is sub-causal for every dagger kernel or causal pure state $f$;
\item \label{enum:scalar-cond}
all scalars $r$ satisfy
\begin{equation} \label{eq:zero-scalar-law}
\scalebox{0.85}{\input{figures/scalar-weird-law.tikz}}
 \ \ 0
\end{equation}
\end{enumerate}
\end{theorem}

\begin{proof}
Appendix~\ref{append:equiv}.
\end{proof}

Note that condition \ref{enum:scalar-cond} is certainly satisfied in any probabilistic theory. Hence any such theory satisfying the remaining conditions is equivalent to $\Quant{\mathbb{R}}$ or $\Quant{\mathbb{C}}$.

\para{Phased rings and fields} A basic open question left from Section~\ref{sec:PureProcessProperties} is the following: is every phased ring $S$ a field? Indeed in any such $S$ every element of the form $1 + s^{\dagger} \cdot s$ is invertible, our motivating examples $\mathbb{R}$ and $\mathbb{C}$ are indeed fields, and it is only in this case that we showed that $\Quant{S}$ satisfies our assumptions (Example~\ref{example:MatofPhasedField}). It is easy to see that when all non-zero scalars in our theory $\catC$ are invertible (as in a probabilistic theory) the phased ring $S$ of scalars in $\plusI{\catC_\pure}$ will be a field; for more on this see Appendix~\ref{sec:phasedFields}.

\para{Daggers}
 Our reconstruction made extensive use of the dagger, which lacks a clear physical interpretation for non-pure processes. It would thus be desirable to avoid explicit use of the dagger, or instead derive its presence as well as dagger compactness from principles such as purification as is essentially done in~\cite{PhysRevA.84.012311InfoDerivQT}. This would yield a fully operational reconstruction in the simple framework of symmetric monoidal categories with discarding. The removal of dagger compactness would also allow infinite-dimensional systems to be considered.
Finally, weakening our principles to allow for super-selection rules would allow us to study the broader category of finite-dimensional C*-algebras and completely positive maps~\cite{coecke2014categories}.

\para{Acknowledgements}
Early stages of this work were discussed during visits to the University of Pavia in June 2016 and to Radboud University Nijmegen in February 2017 and I thank my hosts Giacomo Mauro D'Ariano, Paolo Perinotti and Bart Jacobs. Thanks also to Bob Coecke, Chris Heunen and John van de Wetering for useful feedback. This research forms a part of the author's DPhil thesis, and was supported by EPSRC Studentship {OUCL/2014/SET}.

\bibliographystyle{alpha}
\bibliography{thesis-bib}

\appendix

\section{Omitted Proofs} \label{appendix:Proofs}

\begin{proof}[Proof of Lemma~\ref{lem:IplusIgivesaplusa}]
This is essentially~\cite[Lemma 7.2]{superpos}. Consider such an object $I \pbiprod I$, say with coprojections $\ket{0}, \ket{1}$. For any object $A$, we claim that the object $A \otimes (I \pbiprod I)$ forms a phased dagger biproduct $A \pbiprod A$ with coprojections
\[
\scalebox{0.85}{\input{figures/coproj.tikz}}

\]
for $i=0,1$. First, note that for any morphism $f$ into another object $B$ we have
\[
\scalebox{0.85}{\input{figures/bend-ph-alt.tikz}}

\]
for $i = 0, 1$. From this it is straightforward to see that $A \otimes (I \pbiprod I)$ satisfies the first condition of a phased biproduct, with induced morphisms being unique up to a phase, and that every phase is of the form $\id{A} \otimes u$ for some phase $u$ of $I \pbiprod I$. This makes phases closed under the dagger, so that $I \pbiprod I$ is a phased dagger biproduct as required.
\end{proof}

\begin{proof}[Proof of Lemma~\ref{lem:kernelsGiveBiproducts}]
The condition is necessary since for any phased dagger biproduct $A \pbiprod B$ it may be seen that $\pcoproj_A$ and $\pcoproj_B$ are dagger kernels with $\pcoproj_A = {\pcoproj_B}^\bot$. 

Conversely, suppose that the condition holds and let $k \colon A \to C$ and $l \colon B \to C$ be orthogonal kernels. 
Let $f$ be any endomorphism of $C \pbiprod C$ with $f \circ \pcoproj_1 = \pcoproj_1 \circ k \circ k^{\dagger}$ and $f \circ \pcoproj_2 = \pcoproj_2 \circ l \circ l^{\dagger}$ and let $i := \img(f)$. Then since 
\[
\coker(f) \circ \pcoproj_1 \circ k
=
\coker(f) \circ f \circ \pcoproj_1 \circ k 
= 
0
\]
and similarly for $\pcoproj_2$ and $l$, there are unique $\coproj_A, \coproj_B$ making the following commute:
\[
\begin{tikzcd}
A \rar{k} \arrow[drr,swap, dashed, "\pcoproj_A"]& C \rar{\pcoproj_1} & C \pbiprod C 
& C \lar[swap]{\pcoproj_2} & B \lar[swap]{l} \arrow[dll, dashed,"\pcoproj_B"] 
\\
& & \Img(f) \uar[swap]{i} & & 
\end{tikzcd}
\]
We claim that $\pcoproj_A$ and $\pcoproj_B$ make $\Img(f)$ a phased biproduct $A \pbiprod B$. To see the existence property, given $g \colon A \to D$ and $h \colon B \to D$, let $j \colon C \pbiprod C \to D$ with $j \circ \pcoproj_1 = g \circ k^{\dagger}$ and $j \circ \pcoproj_2 = h \circ l^{\dagger}$. Then the morphism $a := j \circ i$ has $a \circ \pcoproj_A = g$ and $a \circ \pcoproj_B = h$.

We now show the uniqueness property. First, it is straightforward to show that $f^{\dagger}$ has the same composites with $\pcoproj_1$ and $\pcoproj_2$ as $f$. Then $f^{\dagger} \circ \pcoproj_1 \circ k^{\bot} = 0$ and so since $i = \img(f)$ we have $i^{\dagger} \circ \pcoproj_1 \circ k^{\bot} = 0$. Then since $k = \img(k)$ we have
\begin{equation*} 
i^{\dagger} \circ \pcoproj_1 
=
i^{\dagger} \circ \pcoproj_1 \circ k \circ k^{\dagger} 
=
i^{\dagger} \circ i \circ \pcoproj_A \circ k^{\dagger} 
=
\pcoproj_A \circ k^{\dagger} 
\end{equation*}

Now suppose that some morphisms $m, p \colon \Img(f) \to D$ each have the same composites with $\pcoproj_A$ and $\pcoproj_B$. Let $q = m \circ i^{\dagger}$ and $r = p \circ i^{\dagger}$. Then 
\begin{align*}
q \circ \pcoproj_1 
 = 
m \circ i^{\dagger} \circ \pcoproj_1
 =
m \circ \pcoproj_A \circ k^{\dagger}
=
p \circ \pcoproj_A \circ k^{\dagger}
=
r \circ \pcoproj_1 
\end{align*}
and similarly for $\pcoproj_2$. So there is a phase $U$ on $C \pbiprod C$ with $q = r \circ U$. Then $m = p \circ u$ where 
\begin{equation} \label{eq:phonIm}
u = i^{\dagger} \circ U \circ i
\end{equation}
One may verify from the definitions of $\pcoproj_A$ and $\pcoproj_B$ that any such endomorphism $u$ preserves them, establishing the uniqueness property. Moreover, running the above argument with $p = \id{}$ shows that any phase on $\Img(f)$ is of the form~\eqref{eq:phonIm}. In particular, since phases on $C \pbiprod C$ are closed under the dagger, so are those on $\Img(f)$.
\end{proof}

\begin{lemma} \label{lem:dag-ker-coproj}
Let $\catB$ be a dagger category with a phased dagger biproduct $A \pbiprod B$. Then $\pcoproj_A$ and $\pcoproj_B$ are dagger kernels with $\pcoproj_A = {\pcoproj_B}^{\bot}$.
\end{lemma}
\begin{proof}
By definition both morphisms are isometries. Let $\pproj_A := \coproj_A^\dagger$ and $\pproj_B := \coproj_B^\dagger$, and let us show that $\pcoproj_A = \ker(\pproj_B)$. Suppose that $f \colon C \to A \pbiprod B$ has $\pproj_B \circ f = 0$, and let $g = \coproj_A \circ \pproj_A \circ f$. Then $\pproj_A \circ g = \pproj_A \circ f$ and $\pproj_B \circ f = 0 = \pproj_B \circ f$. Hence for some phase $U$ we have 
\[
f 
= 
U \circ g 
=
U \circ \coproj_A \circ \pproj_A \circ f
=
 \coproj_A \circ \pproj_A \circ f
\]
and so $f$ factors over $\coproj_A$, as required. 
\end{proof}

\begin{proof}[Proof of Example~\ref{example:MatofPhasedField}]
We've seen that $\Mat_S$ is always dagger compact with dagger biproducts. 

First let us establish dagger normalisation. 
For any state $\psi = (a_i)^n_{i=1} \colon 1 \to n$, since $S$ is a phased ring we have 
\[
\psi^\dagger \circ \psi = \sum^n_{i=1} a_i^\dagger \cdot a_i = a^\dagger \cdot a
\] 
for some $a \in S$. Then if $\psi \neq 0$ also $a \neq 0$ and so $\phi = (\frac{a_i}{a})^n_{i=1}$ is an isometry with $\psi = \phi \circ a$. We now show that $\Mat_S$ has dagger kernels. Note that the states on any object $n \in \mathbb{N}$ form the vector space $S^n$ and also come with the `inner product' 
\[
\langle \psi, \phi \rangle := \psi^{\dagger} \circ \phi
\] 
for $\psi, \phi \colon 1 \to n$. Since $S$ is a phased ring this satisfies $\langle \psi, \psi \rangle = 0 \implies \psi = 0$. 

Now for any morphism $M \colon n \to m$, the set $\{ \psi \mid M \circ \psi = 0 \}$ is a subspace of $S^n$ and so has a finite basis $\{\psi_i\}^r_{i=1}$ for some $r \leq n$. Using the well-known Gram-Schmidt algorithm (see e.g.~\cite[p.544]{cheney2009linear}) we may replace this by another basis $\{\phi_i\}^r_{i=1}$ which is orthonormal in that $\langle \phi_j^{\dagger}, \phi_i \rangle = \delta_{i,j}$. Then $k := (\phi_i)^r_{i=1} \colon r \to n$ in $\MatS$ is an isometry with $k = \ker(M)$. 

Next we verify homogeneity. Let $M, N \colon n \to m$ satisfy $M^{\dagger} \circ M = N^{\dagger} \circ N$. It follows that $\coim(M) = \coim(N)$ and so after restricting along these we may assume that $\ker(M) = \ker(N) = 0$. Now define a modified `inner product' by 
\[
\langle \psi, \phi \rangle' := \langle M \circ \psi, M \circ \phi \rangle = \langle N \circ \psi, N \circ \phi \rangle
\] 
Again this satisfies $\langle \psi, \psi \rangle' = 0 \implies \psi = 0$. Hence we may again apply the Gram-Schmidt algorithm to find an orthonormal basis $\{e_i\}^n_{i=1}$ with respect to $\langle - , - \rangle'$. Then $\{M \circ e_i\}^n_{i=1}$ and $\{N \circ e_i\}^n_{i=1}$ are each orthonormal collections of states of $m$ and so may be extended to orthonormal bases $\{\psi_i\}^m_{i=1}$ and $\{\phi_i\}^m_{i=1}$ respectively. Finally, any matrix $U \colon m \to m$ with $U \circ \psi_i = \phi_i$ for all $i$ is unitary and satisfies $U \circ M = N$.
\end{proof}

\begin{proposition} \label{prop:quant-to-theory}
Let $\catA$ be a non-trivial quantum category. Then $\CPM(\catA)$ forms a dagger theory satisfying the operational principles. Moreover there is a dagger monoidal equivalence
\begin{equation} \label{eq:Aequiv}
\catA \simeq \plusI{\CPM(\catA)_\pure}
\end{equation}
\end{proposition}
\begin{proof}
By homogeneity in $\catA$, $\CPM(\catA)$ has essentially unique dilations with respect to its environment structure $\Dbl{\catA}$, i.e.~the image of the functor $\Dbl{(-)} \colon \catA \to \CPM(\catA)$. Now since $\catA$ has dagger kernels, it satisfies  
\begin{equation} \label{eq:rule}
f^{\dagger} \circ f = 0 \implies f = 0
\end{equation}
for all morphisms $f$, by~\cite[Lemma 2.4]{vicary2011categorical}. 
It follows that $\Dbl{f} = 0 \implies f = 0$ in $\catA$ and  $\discard{} \circ g = 0 \implies g = 0$ in $\CPM(\catA)$.

The former provides $\Dbl{\catA}$ with dagger kernels given by $\ker(\Dbl{f}) = \Dbl{\ker(f)}$. It follows that in $\CPM(\catA)$ any morphism $g$ with dilation $\Dbl{f}$ has a dagger kernel $\ker(g) = \ker(\Dbl{f})$. Indeed if $g \circ h = 0$ for some $h$ with dilation $\Dbl{j}$ as below then 
\[
0 \ \ 
\scalebox{0.85}{\input{figures/kernel-argument.tikz}}

\text{ and so }
\quad
\scalebox{0.85}{\input{figures/kernel-argument2.tikz}}

\ \  0 
\]
so that $\Dbl{j}$ factors over $\ker(\Dbl{f}) \otimes \id{E}$ and hence $h$ factors over $\ker(\Dbl{f})$. Hence $\CPM(\catA)$ has dagger kernels. 

To show that $\CPM(\catA)$ has \strongpurification, we need to show that a morphism belongs to $\Dbl{\catA}$ iff it is pure. By Lemma~\ref{lem:EUP-properties} and compactness it suffices to show that, for any state $\rho$ and causal state $\sigma$, if $\rho \otimes \sigma \in \Dbl{\catA}$ then so does $\rho$. So suppose that this holds. It follows from well-pointedness in $\catA$ and the rule \eqref{eq:rule} that there is some effect $\psi \in \catA$ for which 
\begin{equation} \label{eq:stateinA}
\scalebox{0.85}{\input{figures/statesinA-alt.tikz}}
 \in \Dbl{\catA}
\end{equation}
is non-zero. Now by dagger normalisation and Proposition~\ref{prop:quant-cat-properties}~\eqref{enum:dagmonodagkernel}, in $\catA$ every state is of the form $k \circ s$ for some dagger kernel state $k$ and scalar $s$. So suppose that the state \eqref{eq:stateinA} takes this form for some $k, s$. Since $\Dbl{k}$ is again a kernel in $\CPM(\catA)$ it follows from zero cancellativity that $\img(\rho) = \Dbl{k}$, and so $\rho = \Dbl{k} \circ t$ for some scalar $t$. Now dagger normalisation in $\catA$ states that every scalar in $\CPM(\catA)$ belongs to $\Dbl{\catA}$. In particular so does $t$ and hence so does $\rho$, as required. Hence $\CPM(\catA)$ has \strongpurification.

Now we've seen that all scalars are pure, giving $\CPM(\catA)$ normalisation, and by the CP axiom and Proposition~\ref{prop:quant-cat-properties}~\eqref{enum:dagmonodagkernel} every causal pure state in is a kernel. Hence by Lemma~\ref{lem:exclusion-equiv-new} $\CPM(\catA)$ satisfies pure exclusion.

Next we show that non-triviality of $\catA$ ensures non-triviality of the dagger theory $\CPM(\catA)$. Let $A=I \biprod I$ in $\catA$. Then if $\discard{A}$ is an isomorphism in $\CPM(\catA)$ it is pure and hence unitary, giving a unitary $\psi = [a,b] \colon I \biprod I \to I$ in $\catA$. But $\psi$ being an isometry is equivalent to $a$ and $b$ being unitary scalars in $\catA$ with $a^{\dagger} \cdot b = 0$. But then $a= b =0$ and so $\id{I} = 0$, a contradiction.

Now, we've seen that the addition operation $+$ in $\catA$ provides $\CPM(\catA)$ with a coarse-graining operation. Moreover thanks to Proposition~\ref{prop:quant-cat-properties}~\eqref{enum:decomp-property} in $\catA$ all dagger kernels $k \colon K \to A$ satisfy 
\[
\scalebox{0.85}{\input{figures/kersumrule.tikz}}

\]
which translates precisely to~\eqref{eq:causally-complemented} in $\CPM(\catA)$. Hence by Proposition~\ref{prop:coarse-graining} $\CPM(\catA)$ satisfies the remaining operational principles.

Finally, we turn to the equivalence \eqref{eq:Aequiv}. 
First, choose as global phases all unitary scalars in $\catA$, so that as before we write $f \sim g$ when $f = u \cdot g$ for some unitary scalar $u$ and $\catA_\sim$ for the quotient category under this congruence. Then by a general result from~\cite{superpos} (being the converse to Theorem~\ref{thm:onPHConstr}) we have $\catA \simeq \plusI{\catA_{\sim}}$. On the other hand by homogeneity we in fact have $\Dbl{f} = \Dbl{g} \iff f \sim g$ since:
\[
\scalebox{0.85}{\input{figures/Dbl-argument-new.tikz}}

\]
for some unitary scalar $u$. But since $\Dbl{\catA} = \CPM(\catA)_\pure$, we obtain a dagger monoidal equivalence $\catA_\sim \simeq \CPM(\catA)_{\pure}$ and hence also the equivalence \eqref{eq:Aequiv}.
\end{proof}

\begin{proof}[Proof of Lemma \ref{lem:prob}]
Clearly $\Rpos$ satisfies these properties. Conversely if they hold then, by Lemmas~\ref{lem:square roots} and~\ref{lem:its_a_field}, $\dring{R}$ is a totally ordered Archimedean field~\cite{hall2011completeness} with $R$ as its positive elements. Let $(x_n)^{\infty}_{n=1}$ be any bounded monotonic sequence in $\dring{R}$. Then for some $r \in R$ and $t=\pm 1$, the bounded sequence $(t x_n + r)^\infty_{n=1}$ is increasing and belongs to $R$, and so converges there. Hence $(x_n)^{\infty}_{n=1}$ also converges in $D(R)$, making the latter monotone complete. But then by~\cite[Theorem 3.11]{hall2011completeness} there is an isomorphism $D(R) \simeq \mathbb{R}$ and hence $R \simeq \Rpos$.
\end{proof}

\section{Notions of Purity} \label{appendix:pure}

There have been several proposed notions of `pure' morphism in the recent literature. The definition used here~\eqref{eq:whole} from~\cite{chiribella2014distinguishability} coincides with that of~\cite{selby2017leaks} whenever the `no leaks' condition~\eqref{eq:noleaks} is satisfied. In fact it is automatic in any theory with a suitable form of essentially unique purification, in the following sense. 

\begin{lemma} \label{lem:EUP-properties}
Let $\catC$ be a dagger theory with a zero object. Suppose that $\catC$ has an environment structure $\catC_\prepure$ which contains all zero morphisms and is state-inhabited, and for which dilations are essentially unique in the sense of~\eqref{eq:essential-uniqueness}. Suppose further that whenever $f \otimes \rho \in \catC_\prepure$ for some causal state $\rho$, so does $f$. Then a morphism belongs to $\catC_\prepure$ iff it is pure.  
\end{lemma}

\begin{proof}
Let $f \colon A \to B$ be non-zero and belong to $\catC_{\prepure}$. Suppose that $f$ has a dilation $g \colon A \to B \otimes C$. Dilating $g$ if necessary, we may assume that $g \in \catC_{\prepure}$. Since $f$ is non-zero, $C$ is not a zero object and so has a causal state $\psi \in \catC_\prepure$. Then 
\[
\scalebox{0.85}{\input{figures/whole_arg3.tikz}}

\quad
\text{ and so }
\quad
\scalebox{0.85}{\input{figures/whole_arg4.tikz}}

\]
for some unitary $U \in \catC_\prepure$. By the CP axiom $U$ is causal and hence so is the state $U \circ \psi$ as desired. Conversely, suppose that $f \colon A \to B$ is pure and non-zero. Let $g \colon A \to B \otimes C$ be a dilation of $f$ with $g \in \catC_{\prepure}$. Then we have
\[
\scalebox{0.85}{\input{figures/whole_arg_new1.tikz}}

\]
for some causal state $\sigma$ of $C$. Hence by assumption $f \in \catC_\prepure$. 
\end{proof}

A more typical notion of purity for probabilistic theories instead refers to a process having no non-trivial decompositions as a mixture. Cunningham and Heunen have also presented an alternative, categorically well-behaved notion of purity~\cite{cunningham2017purity}. Thanks to the presence of coarse-graining $+$, all of these notions may be defined in theories satisfying our principles, and in fact they coincide. 

\begin{lemma}
Let $\catC$ be a dagger theory satisfying the operational principles. For any morphism $f \colon A \to B$ the following are equivalent:
\begin{enumerate}
\item \label{enum:pure}
$f$ is pure in the sense of~\eqref{eq:whole};
\item \label{enum:atomic}
$f$ is \emph{atomic}~\cite{PhysRevA.84.012311InfoDerivQT}: $f = g + h \implies g = r \cdot f$ for some scalar $r$;
\item \label{enum:OscarPure}
$f$ is \emph{copure} in the sense of~\cite{cunningham2017purity}:
\begin{align*}
\scalebox{0.85}{\input{figures/copure.tikz}}
\\ 
 \text{ for some $k$ with } \quad   
\scalebox{0.85}{\input{figures/copure2.tikz}}

\end{align*}
\end{enumerate}
\end{lemma}
\begin{proof}
\ref{enum:pure} $\implies$ \ref{enum:atomic}:
Suppose that $f$ is pure and $f = g + h$ for some morphisms $g, h$. 
Let $C$ be any object with a pair of orthonormal pure states $\ket{0}, \ket{1}$, and let $k = g \otimes \ket{0} + h \otimes \ket{1}$. Then $k$ is a dilation of $f$ and so it is simply given by tensoring $f$ with some causal state $\rho$. But then 
\[
\scalebox{0.85}{\input{figures/whole_variation3.tikz}}

\]

\ref{enum:atomic} $\implies$ \ref{enum:pure}:
Let $p \colon A \to B \otimes C$ be a purification of $f$. If $C = 0$ then $f = 0$, otherwise for any pure causal state $\psi$ of $C$ we have $\psi^{\dagger} + e = \discard{C}$ for some effect $e$ by~\eqref{eq:causally-complemented}. But then since $f$ is atomic we have implications
\[
\scalebox{0.85}{\input{figures/atomic-arg.tikz}}

\quad \implies \quad 
\scalebox{0.85}{\input{figures/atomic-arg2.tikz}}

\]
for some scalar $r$. Then since $p$ and $\psi$ is pure so is $r \cdot f$. Now if $f$ is non-zero then so is $p$ and hence there is some pure state $\psi$ of $C$ which belongs to the image of $p$, and so makes the above and the scalar $r$ non-zero. But if $r \neq 0$ and $r \cdot f$ is pure then so is $f$ by Lemma~\ref{lem:exclusion-equiv-new}~\eqref{enum:puretensor}.


\ref{enum:OscarPure} $\implies$ \ref{enum:pure}: thanks to the `no leaks' property~\eqref{eq:noleaks}. 

\ref{enum:pure} $\implies$ \ref{enum:OscarPure}: Let $f$ be pure and suppose the left hand equation is satisfied. Thanks to purification we can assume that $h$ is pure. By bending wires it suffices to consider when $D = I$. Let $l \colon B \otimes C \to F$ purify $g$. If $E$ or $F$ are zero objects then $h = 0$. Otherwise let $\psi$ and $\phi$ be causal pure states of $E, F$ respectively. Then we have
\begin{align*}
\scalebox{0.85}{\input{figures/copurearg.tikz}}
\\
\text{ so } \quad
\scalebox{0.85}{\input{figures/copurearg2.tikz}}

\end{align*}
for some unitary $U$ by essential uniqueness, with the morphism above $f$ on the right-hand side being a dilation of $g$ as required.
\end{proof}

\section{Equivalent Axioms} \label{append:equiv}

We have seen that the operational principles imply the presence of a coarse-graining operation $+$ on processes. Taking this operation instead as primitive allows us to reformulate our principles in an alternative, operationally motivated manner. First let us note that this operation is surprisingly well-behaved. 

\begin{proposition} \label{cor:cancellative}
Let $\catC$ be a dagger theory satisfying the operational principles. Then the following hold for all morphisms $f, g, h$: 
\begin{itemize}
\item
$f + g = f + h \implies g = h$;
\item
$f \otimes g = f \otimes h \implies f = 0 \text{ or } g = h$.
\end{itemize}
\end{proposition}
\begin{proof}
We have $\catC \simeq \CPM(\catA)$ for a quantum category $\catA$. But the definition of coarse-graining in $\CPM(\catA)$ is simply addition in $\catA$. Since $\catA$ has negatives $-f$ for all morphisms $f$ it satisfies both properties. 
\end{proof}
\noindent
Under even milder assumptions we can deduce another property of coarse-graining.

\begin{lemma} \label{lem:OpLemma-Redux}
Let $\catC$ be a non-trivial compact dagger theory with coarse-graining and dagger kernels, satisfying \strongpurification{} and pure exclusion. Then all kernels are causal and $f + g = 0 \implies f = g = 0$ for all morphisms $f, g \colon A \to B$. 
\end{lemma}
\begin{proof}
Noting that causal complementation of kernels was not required for its proof, by Proposition~\ref{prop:OpTheoryProperties} all kernels are pure isometries and hence causal by the CP axiom, and any non-trivial object $C$ has an orthonormal pair of pure states $\ket{0}, \ket{1}$. 

Now suppose that $f, g \colon A \to B$ satisfy $f + g = 0$. Then
\[
\scalebox{0.85}{\input{figures/pos-arg.tikz}}

\quad 
\text{ has }
\quad 
\scalebox{0.85}{\input{figures/pos-arg2.tikz}}

\ 0 
\]
and so $h = 0$ by the same proposition. But then applying $\ket{0}$ we obtain $f = 0$, and similarly $g = 0$ also.
\end{proof}

Conversely, if we instead assume the presence of a coarse-graining operation satisfying some natural assumptions, several of our principles in fact follow automatically.

\begin{lemma} \label{lem:caus-complemented}
Let $\catC$ be a dagger theory with coarse-graining satisfying
\begin{align*}
\discard{} + e = \discard{} &\implies e = 0 \label{eq:effect-zero-law}
\\ 
d + e = 0 &\implies d = e = 0
\end{align*}
for all effects $d, e$. 
Then all kernels satisfy~\eqref{eq:causally-complemented} iff all kernels and cokernels are sub-causal. Hence in this case they are causally complemented.
\end{lemma}

\begin{proof}
The equation~\eqref{eq:causally-complemented} makes all cokernels sub-causal, and composing with any kernel $k$ gives that $k$ is causal. Conversely let $k \colon K \to A$ be a kernel for which $k$ and $k^\dagger$ are sub-causal, say with 
\[
\scalebox{0.85}{\input{figures/kersc.tikz}}
  \qquad \qquad 
\scalebox{0.85}{\input{figures/kersc2.tikz}}

\]
for some effects $a, b$. Then since $k$ is an isometry we obtain $\discard{K} = \discard{K} + b \circ k + a$ and so $b \circ k = 0 = 0$. Hence all kernels are causal. It follows that $c = k^{\bot \dagger}$ has
\begin{align*}
\discard{K^\bot} \circ c
&= \discard{A} \circ c^\dagger \circ c  & (\text{c causal}) \\ 
&= b \circ c^\dagger \circ c  & (c \circ k = 0) \\  
&= b  &  (b \circ k = 0)
\end{align*}
as required. The final statement is from Proposition~\ref{prop:coarse-graining}.
\end{proof}

Next, pure exclusion can also be deduced easily.

\begin{lemma} \label{lem:orthogonality}
Let $\catC$ be a dagger theory with \strongpurification, dagger kernels, normalisation, and coarse-graining satisfying \eqref{eq:zero-scalar-law}
for all scalars $r$. Suppose also that $\psi^{\dagger}$ is sub-causal for every causal pure state $\psi$. Then $\catC$ satisfies pure exclusion. 
\end{lemma}
\begin{proof}
By Lemma~\ref{lem:exclusion-equiv-new} it remains to show that every causal pure state $\psi \colon I \to A$ is a kernel. It suffices to assume $\coker(\psi) = 0$ and show that $\psi$ is unitary. Then there is an effect $e$ for which 
\[
\scalebox{0.85}{\input{figures/pureexcl1.tikz}}
 \quad \text{ and so } \quad 
\scalebox{0.85}{\input{figures/pureexcl12.tikz}}

\] 
since $\psi$ is a causal isometry. Hence $e \circ \psi = 0$ and so $e = 0$ giving $\psi^{\dagger} = \discard{A}$. But then by essential uniqueness $\psi \circ \psi^{\dagger}$ is unitary, making $\psi$ unitary also.
\end{proof}

We can now present our principles in the equivalent manner of Theorem~\ref{thm:alternate-axioms}.

\begin{proof}[Proof of Theorem~\ref{thm:alternate-axioms}]
If $\catC$ satisfies the principles then the first point holds by Corollary~\ref{cor:cancellative}. If $k$ is a kernel then $k^{\dagger}$ is sub-causal by~\eqref{eq:causally-complemented}. In particular if $\psi$ is a causal pure state then $\psi^{\dagger}$ is sub-causal.

Conversely, if these hold then by Lemma~\ref{lem:orthogonality} pure exclusion holds, and so by Lemmas~\ref{lem:caus-complemented} and ~\ref{lem:OpLemma-Redux} it remains to check that $\discard{} + e = \discard{} \implies e = 0$ for all effects $e$. But if this equation is satisfied then we get that $e \circ \rho = 0$ for any causal state $\rho$. In particular, since kernels are causal, for any causal pure state $\psi$ of $\CoIm(e)$ we have $e \circ \coim(e)^{\dagger} \circ \psi = 0$ and so $\psi = 0$. Hence $\CoIm(e) = 0$ and so $e = 0$. 
\end{proof}

\section{Pre-Quantum Categories} \label{sec:pre-quantum}

For any theory $\catC$ satisfying the operational principles, $\catC_\pure$ may be axiomatised as follows.

\begin{definition}
A \emph{pre-quantum} category $\catB$ is a dagger compact category with the superposition properties and dagger kernels, satisfying dagger normalisation and homogeneity as in Definition \ref{def:quantumcat}, and for which $\id{I}$ is its only unitary scalar. We again call a pre-quantum category \emph{non-trivial} when $\id{I} \neq 0$.
\end{definition}

In fact pre-quantum categories can be equivalently defined under a weakening of \eqref{eq:positive-cancellation} called `positive-freeness'; see \cite[Chapter 6]{thesis}.

\begin{proposition}
The correspondence of Theorem~\ref{thm:mainToQuantumCat} extends to include non-trivial pre-quantum categories $\catB$, via $\catB = \catC_\pure$ and $\catA=\plusI{\catB}$.
\end{proposition}
\begin{proof}
Firstly, if $\catC$ satisfies the operational principles, we must check that $\catC_\pure$ is a pre-quantum category. But this is immediate from the results of Sections \ref{sec:consequences} and \ref{sec:derivingsups} and strong purification.

Next, the proof that if $\catB$ is a pre-quantum category then $\catA := \plusI{\catB}$ is a quantum category is almost exactly that of Proposition \ref{prop:QuantWOGToQuantBiprod} (replacing $\catC_\pure$ by $\catB$). Observe also that if $\catB$ is non-trivial then so is $\catA$. Finally, for the correspondence,  note that in this case $\catB \simeq \catA_\sim \simeq \CPM(\catA)_{\pure}$, with the latter equivalence found in the proof of Theorem~\ref{thm:mainToQuantumCat}.
\end{proof}

\section{Phased Rings and Fields} \label{sec:phasedFields}

We have left open the question of whether every phased ring $S$ is in fact a field. Two results in this direction are as follows. 

\begin{lemma} \label{lem:no_nontriv_subob}
Let $\catA$ be a quantum category, $S$ its ring of scalars and $R = S^{\pos}$. 
The following are equivalent:
\begin{enumerate} 
\item \label{enum:R-divis}
$R$ is a semi-field;
\item \label{enum:S-divis}
$S$ is a field;
\item \label{enum:Subob}
In $\catA$ the only sub-objects of $I$ are $\{0, I\}$. 
\end{enumerate}
\end{lemma}
\begin{proof}
\ref{enum:R-divis} $\iff$ \ref{enum:S-divis}: $R \subseteq S$ and an element $s \in S$ is invertible iff $s^{\dagger} \cdot s$ is. 

\ref{enum:S-divis} $\implies$ \ref{enum:Subob}:
Let $m \colon A \rightarrowtail I$ be monic. Then if $r: = m \circ m^\dagger$ is zero then $m = 0$ by~\eqref{eq:rule}. Otherwise $r$ is invertible and hence so is $m$. 

\ref{enum:Subob} $\implies$ \ref{enum:S-divis}: 
Thanks to  Lemma~\ref{lem:dag-ker-in-compact} every non-zero scalar $r$ has $\ker(r) = 0$. Since $\catA$ has negatives by Proposition~\ref{prop:quant-cat-properties}, this makes $r$ monic and hence an isomorphism by assumption. 
\end{proof}

Recall the order $\leq$ on $S^{\pos}$ from Lemma~\ref{lem:square roots}. 

\begin{lemma} \label{lem:its_a_field}
Let $S$ be a phased ring and suppose that every $a \in S^{\pos}$ has $\frac{1}{n} \leq a$ for some $n \in \mathbb{N}$. Then $S$ is a field.
\end{lemma}
\begin{proof}
From the definition of a phased ring, we have that for $a, b \in S^{\pos}$ whenever $a \leq b$ then $a$ is divisible by $b$. If $a \neq 0$ then $1 \leq a \cdot n$ for some $n \in \mathbb{N}$, making $a \cdot n$ invertible and hence $a$ also. Then every $s \in S$ is invertible since $s^{\dagger} \cdot s \in S^\pos$ is. 
\end{proof}

\section{Hilbert Categories} \label{apped:Hilbert}

Our notion of a quantum category is similar to the following axiomatization of the category $\Hilb$ due to Chris Heunen~\cite{heunen2009embedding}. A \emph{Hilbert category} is a dagger symmetric monoidal category $\catH$ with dagger equalisers and finite dagger biproducts such that:
\begin{itemize}
\item every isometry is a kernel;
\item every morphism $f \colon A \to B$ has a \emph{bound} in that there is some scalar $s$ such that for every $\psi \colon I \to A$ we have $\psi^{\dagger} \circ f^{\dagger} \circ f \circ \psi = (s^{\dagger} \circ s) \circ (\psi^{\dagger} \circ \psi) + r$ for some positive scalar $r$.
\end{itemize}

\begin{lemma} \label{lem:QuantIsHilbert}
Let $\catA$ be a quantum category. Then $\catA$ is a well-pointed Hilbert category.
\end{lemma}

\begin{proof}
By Proposition~\ref{prop:quant-cat-properties} it remains to verify that every morphism $f \colon A \to B$ has a bound $s$ as above. Thanks to dagger normalisation it suffices to consider when $\psi$ is an isometry and hence a kernel. Then letting $c = \coker(\psi)$ by Proposition~\ref{prop:quant-cat-properties}~\eqref{enum:decomp-property} we have $\id{A} = \psi \circ \psi^{\dagger} + c^\dagger \circ c$ so that:
\[
\scalebox{0.85}{\input{figures/Trace-pic.tikz}}

\]
since the right-hand scalar is positive, we may take $s$ as the left-hand scalar.
\end{proof}

By~\cite[Theorem~4]{heunen2009embedding} this means that whenever $\catA$ is a quantum category which is locally small and has that its ring of scalars $S$ is a field (as in Lemma~\ref{lem:no_nontriv_subob}) of at most continuum cardinality, there is a lax dagger monoidal embedding $\catA \hookrightarrow \Hilb$ up to some isomorphism of $S$. It would be interesting to further explore the connections between our result and Heunen's.

\end{document}

\end{document}